\documentclass[a4paper,12pt]{scrartcl}
\usepackage{amsmath,amssymb,latexsym,amsthm,enumerate}
\usepackage{xcolor}
\usepackage{graphicx}
\usepackage{epstopdf}
\usepackage{changes}
\usepackage{hyperref}

\usepackage[T1]{fontenc}
\usepackage[utf8]{inputenc}
\usepackage[english]{babel}

\newcommand*{\ol}{\overline}

\newcommand*{\bbQ}{\mathbb Q}
\newcommand*{\bbR}{\mathbb R}
\newcommand*{\bbZ}{\mathbb Z}

\newcommand*{\cF}{\mathcal{F}}

\newcommand*{\cN}{\mathcal{N}}

\newcommand*{\cS}{\mathcal{S}}

\newcommand*{\N}{\mathbb{N}}

\newcommand*{\Z}{\mathbb{Z}}

\newcommand*{\R}{\mathbb{R}}

\DeclareMathOperator*{\essinf}{ess\,inf}

\newcommand*{\Linfm}{L^{\infty-}} 
\newcommand*{\Ltwop}{L^{2+}} 

\newcommand*{\lbeta}{\bar{\beta}}
\newcommand*{\leta}{\bar{\eta}}
\newcommand*{\lalpha}{\bar{\alpha}}

\newcommand*{\lYtwo}{\bar{y}}
\newcommand*{\lYo}{\bar{Y}^{(o)}} 
\newcommand*{\lYe}{\bar{Y}^{(e)}} 

\newcommand{\be}{\begin{eqnarray*}}
\newcommand{\ee}{\end{eqnarray*}}
\newcommand{\ben}{\begin{eqnarray}}
\newcommand{\een}{\end{eqnarray}}
\newcommand{\bi}{\begin{itemize}}
\newcommand{\ei}{\end{itemize}}

\newtheorem{theo}{Theorem}[section]

\newtheorem{lemma}[theo]{Lemma}
\newtheorem{propo}[theo]{Proposition}
\newtheorem{corollary}[theo]{Corollary}

\theoremstyle{definition}
\newtheorem{ex}[theo]{Example}

\newtheorem{remark}[theo]{Remark}

\title{Optimal trade execution in an order book model with stochastic liquidity parameters}

\author{Julia Ackermann\thanks{Institute of Mathematics, University of Gie{\ss}en, Arndtstr.~2, 35392 Gießen, Germany.
\emph{Email:} julia.ackermann@math.uni-giessen.de, \emph{Phone:} +49 (0)641 9932113.}
\and Thomas Kruse\thanks{Institute of Mathematics, University of Gie{\ss}en, Arndtstr.~2, 35392 Gießen, Germany.
\emph{Email:} thomas.kruse@math.uni-giessen.de, \emph{Phone:} +49 (0)641 9932102.}
\and Mikhail Urusov\thanks{Faculty of Mathematics, University of Duisburg-Essen, Thea-Leymann-Str.~9, 45127 Essen, Germany.
\emph{Email:} mikhail.urusov@uni-due.de, \emph{Phone:} +49 (0)201 1837428.
}}

\definechangesauthor[color=red]{TK}
\definechangesauthor[color=blue]{MU}
\definechangesauthor[color=cyan]{JA}

\begin{document}

\maketitle

\begin{abstract}
We analyze an optimal trade execution problem in a financial market with stochastic liquidity.
To this end we set up a limit order book model in which both order book depth and resilience evolve randomly in time.
Trading is allowed in both directions and at discrete points in time.
We derive an explicit recursion that, under certain structural assumptions, characterizes minimal execution costs.
We also discuss several qualitative aspects of optimal strategies, such as existence of profitable round trips or closing the position in one go,
and compare our findings with the literature.


\end{abstract}

\section*{Introduction}

Market liquidity describes the extent to which buying (resp.\ selling) an asset moves the price against the buyer (resp.\ seller). 
In an illiquid financial market large orders have a substantial adverse effect on the realized prices. 
Typically, this effect is not constant over time. Temporal variations of liquidity are partly driven by deterministic trends such as intra-day patterns. In addition, there exist random changes in liquidity such as liquidity shocks that superimpose the deterministic evolution. To benefit from times when trading is cheap, institutional investors continuously monitor the available liquidity and
schedule their order flow accordingly.
The scientific literature on optimal trade execution problems deals with the optimization of trading schedules, when an investor faces the task of closing a position in an illiquid market.
Incorporating random fluctuations of liquidity into models of optimal trade execution constitutes a highly active field of research
(see, e.g.,
\cite{ackermann2020cadlag,
almgren2012optimal,
ankirchner2020optimal,
ankirchner2014bsdes,
ankirchner2015optimal,
bank2018linear,
cheridito2014optimal,
fruth2019optimal,
graewe2017optimal,
graewe2015non,
graewe2018smooth,
horst2016constrained,
horst2019multi,
kruse2016minimal,
popier2019second,
schied2013control}
and references therein; see below for an extended literature discussion).

In this work we analyze a trade execution problem in a financial market model with linear stochastic price impact and stochastic resilience. To be more specific, we consider a block-shaped limit order book, where liquidity is uniformly distributed to the left and to the right of the mid-price. 
To account for stochastic liquidity, the depth of the order book is allowed to vary randomly in time.
At initial time $0$ the investor observes the current order book depth $1/\gamma_0> 0$
but has no precise knowledge about the order book depth at future times (only a probabilistic assessment).
If the investor executes a trade\footnote{We allow for both buy ($\xi\ge 0$) and sell ($\xi\le 0$) orders.} of size $\xi_0\in \R$ at time $0$, 
she incurs costs of size $\frac{\gamma_0}{2} \xi^2_0$. Moreover, the trade of size $\xi_0$ shifts the mid-price of the order book by $\gamma_0 \xi_0$. Observe that this deviation is positive if and only if $\xi_0>0$, i.e., if $\xi_0$ is a buy order. In the period from time $0$ to the next trading time $1$ this deviation changes from 
$\gamma_0 \xi_0$ to $D_{1-}=\beta_1\gamma_0 \xi_0$, where $\beta_1>0$ is a positive stochastic factor (unknown to the investor at time $0$). The factor $\beta_1$ describes the resilience of the order book: if $\beta_1$ is close to $0$ the order book nearly fully recovers from the trade $\xi_0$, whereas if $\beta_1$ is close to $1$ the impact of $\xi_0$ persists. We highlight here that 
we do not exclude the case, where the event $\{\beta_1>1\}$ has positive probability,
which would reflect a possibility of self-exciting behavior of the market impact.
At time $1$ the value of $\beta_1$ is disclosed to the investor. Moreover, she observes the updated order book depth $1/\gamma_1>0$. 
Based on this information the investor executes a trade of size $\xi_1$ which generates costs $(D_{1-}+\frac{\gamma_1}{2}\xi_1)\xi_1$ and moves the deviation to $D_{1-}+\gamma_1\xi_1$. By continuing this sequence of operations to arbitrary trading times $k\in \N$ we thus obtain our financial market model with stochastic price impact (described by a positive process $\gamma=(\gamma_k)_{k\in \N_0}$) and stochastic resilience (described by a positive process $\beta=(\beta_k)_{k\in \N_0}$).

In this financial market we consider an investor who has to close a financial position of size $x\in \R$ up to a given time $N\in \N$. We assume that the investor is risk-neutral and aims at minimizing the overall trading costs. Apart from some technical integrability conditions we do not a priori impose any restrictions on trading strategies of the investor. In particular, even if the task is to sell a certain amount of assets (i.e., $x>0$), we allow for trading strategies where the investor buys assets at some points in time.

\medskip
The above description of the model highlights that our setting is a certain discrete-time formulation within the class of limit order book models,
where the liquidity parameters are stochastic
(i.e., both the price impact and the resilience are positive random processes).
The approach to mathematically model liquidity via order book considerations was initiated in
\cite{alfonsi2008constrained},
\cite{alfonsi2010optimal},
\cite{alfonsi2010boptimal},
\cite{obizhaeva2013optimal} and
\cite{predoiu2011optimal}.
Limit order book models with deterministically time-varying liquidity are studied in
\cite{alfonsi2014optimal},
\cite{bank2014optimal} and
\cite{fruth2014optimal},
while stochastic liquidity is discussed in
\cite{fruth2019optimal}.
We point out the following essential differences between our current setting and the settings in the aforementioned papers.
\begin{enumerate}[(a)]
\item
Both in the present paper and in \cite{fruth2019optimal},
$\beta$ and $\gamma$ are random processes, while they are deterministic functions of time in
\cite{alfonsi2014optimal},
\cite{bank2014optimal} and
\cite{fruth2014optimal}.
\item
In \cite{bank2014optimal},
\cite{fruth2014optimal} and
\cite{fruth2019optimal},
execution strategies are constrained in one direction,
while trading in both directions is allowed
in the present paper and in \cite{alfonsi2014optimal}.
\item
In \cite{alfonsi2014optimal},
\cite{bank2014optimal},
\cite{fruth2014optimal} and
\cite{fruth2019optimal},
the resilience process (or function)
$\beta$ is assumed to be $(0,1)$-valued,
while we only require it to be positive in the present paper.
\end{enumerate}
In our setting we encounter several new qualitative effects,
which are briefly mentioned below and discussed
in more detail in the main body of the paper.
Moreover, for each of these effects, we identify its reason 
by constructing pertinent examples.

We also mention \cite{ackermann2020cadlag}, which is a continuous-time counterpart of our present paper.
Notice that the discrete-time optimization problem considered in the present paper can be embedded into the continuous-time framework of \cite{ackermann2020cadlag},
but then it would become an additionally constrained optimization problem, where the strategies are restricted to trade only at discrete time points.
We stress that \cite{ackermann2020cadlag} does not solve such a problem (it deals with continuous time only), and both papers rather concentrate on studying different questions.
Moreover, \cite{ackermann2020cadlag} does not study the qualitative effects mentioned in the previous paragraph (and discussed below);
instead we need to work with a challenging quadratic backward stochastic differential equation (BSDE) in \cite{ackermann2020cadlag} and extend the continuous-time problem to incorporate execution strategies of infinite variation.
In particular, some of the results of the present paper are required in \cite{ackermann2020cadlag} to derive the appropriate problem formulation and the mentioned quadratic BSDE as continuous-time limits of the corresponding discrete-time objects. 
In this connection it is worth noting that, as explained in Appendix~A of \cite{ackermann2020cadlag} in detail,
when trading frequencies increase, the deviation process and the costs from discrete time converge to
continuous-time counterparts that include additional terms in comparison with the usual formulations.
We also refer to the introduction of \cite{ackermann2020cadlag} for a detailed discussion of the importance of these additional terms.

\medskip
In Theorem~\ref{thm:mainres} we show that the optimal trading strategies and the minimal expected trading costs are characterized by a single stochastic process $Y=(Y_n)_{n\in \Z \cap (-\infty,N]}$ which is defined via a backward recursion. We prove Theorem~\ref{thm:mainres} by means of dynamic programming. To this end we put the trade execution problem into a dynamic framework and allow for arbitrary initial times $n\in \Z \cap (-\infty,N]$, arbitrary initial positions $x\in \R$ and arbitrary initial market deviations $d\in \R$. In this setting we show that the minimal expected overall execution costs amount to
\begin{equation}\label{eq:val_fct_intro}
V_n(x,d)=\frac{Y_n}{\gamma_n}(d-\gamma_n x)^2-\frac{d^2}{2\gamma_n}.
\end{equation}
In particular, for each $n\in \Z \cap (-\infty,N]$ it follows that the random variable $2Y_n$ takes values in $(0,1]$ and describes
to which percentage the costs of closing one unit $x=1$ at time $n$ immediately can be reduced by executing this position optimally over $\{n,\ldots,N\}$ (given no initial market deviation $d=0$). Accordingly, if $Y_n$ is close to $1/2$ it is nearly optimal to close the position immediately in one go, whereas if $Y_n$ is close to $0$ it pays off to split the position and to put only a small fraction in the market at time $n$.

\medskip
In the remainder of the article we discuss several qualitative and quantitative properties of our market model and the trade execution problem. For instance, we analyze whether our financial market admits price manipulation (in the sense of Huberman and Stanzl \cite{huberman2004price},
see also \cite{alfonsi2010boptimal} or \cite{gatheral2010no}).
A financial market is said to admit price manipulation if there exist round trip strategies (i.e., execution strategies that start in the initial position $x=0$) that generate profits in expectation. It follows immediately from \eqref{eq:val_fct_intro} that if there is no initial market deviation (i.e., $d=0$), then the market does not admit price manipulation. However, for general $d\in \R$ we have that
$V_n(0,d)=\frac{d^2}{\gamma_n}(Y_n-\frac{1}{2})$ and thus in the case $d\neq 0$ there exist profitable round trips starting at time $n$ if and only if $Y_n<\frac{1}{2}$. We show that if the investor has a directional view on the resilience process at time~$n$
(i.e., $E[\beta_{n+1}|\mathcal F_n]\neq 1$, where $\cF_n$ represents the information available at time~$n$), then she can exploit the information $d\neq 0$ and construct profitable round trips (see Corollary~\ref{propo:propY} and the subsequent discussion). This is in line with the results in \cite{fruth2014optimal}
and \cite{fruth2019optimal}, where $\beta$ is assumed to take values in $(0,1)$.
Interestingly, in our model profitable round trips with $d\ne0$ can in general exist even on a part of the event $\{E[\beta_{n+1}|\cF_n]=1\}$
and, moreover, even when there is no directional view on the resilience  in all future time points (see Example~\ref{ex:11072019a2}).

A further interesting effect that appears because we do not restrict the process $\beta$ to take values in $(0,1)$ concerns the question under which conditions it is optimal to close the position in one go. 
We notice that in the settings of \cite{fruth2014optimal} and \cite{fruth2019optimal}, where, in particular, $\beta$ is $(0,1)$-valued, it is never optimal to close the position prematurely (see Proposition~A.3 in \cite{fruth2019optimal}).
On the contrary, in our setting, closing the position prematurely can be optimal even with deterministic $\beta$ and $\gamma$ (Example~\ref{ex:01062019a1}) but is never optimal with the additional restriction for $\beta$ to be $(0,1)$-valued (Proposition~\ref{rem:close_immed}).
Moreover, in the situation when closing the position prematurely is optimal, it can either be optimal to build up a new position at the next time point (Example~\ref{ex:01062019a1}) or not to trade any longer (the latter happens on the event $\{Y_n=\frac12\}$, see Proposition~\ref{prop:01062019a1} and Corollary~\ref{propo:propY}).
On the other hand, when we allow for stochastic $\beta$ and $\gamma$, closing the position prematurely can be optimal even with $(0,1)$-valued $\beta$ (Example~\ref{ex:13072019a1}).
We, finally, notice that the difference between the latter statement and the mentioned Proposition~A.3 in \cite{fruth2019optimal}
is due to the fact that, in contrast to our current setting, in \cite{fruth2019optimal} the trading is constrained only in one direction.
We refer to Table~\ref{tab:09062020a1} in Section~\ref{sec:1go} for a more detailed discussion.

Furthermore, we address the question of how much better in comparison to the immediate position closure the investor can perform if the time horizon is very large. That is, we analyze the behavior of the random sequence $(Y_n)_{n\in \{\ldots,N-1,N\}}$ as $n\to -\infty$.
If liquidity increases on average (more precisely, if $\gamma$ is a supermartingale) we show that $(Y_n)_{n\in \{\ldots,N-1,N\}}$ converges
a.s.\ and in any $L^p$, $p\in[1,\infty)$, to a $[0,1/2]$-valued random variable as $n\to -\infty$ (Proposition~\ref{prop:08062019a1}). If liquidity decreases on average, then, in general, the limit can fail to exist (Lemma~\ref{lemma:Ydoesnotconverge}). In a more specific setting, where the multiplicative increments of the price impact $\eta_{k+1}=\gamma_{k+1}/\gamma_k$ and the resilience factor $\beta_{k+1}$ are independent of the history up to time $k$ and their expectations are homogeneous in time, the limit of $(Y_n)_{n\in \{\ldots,N-1,N\}}$ as $n\to-\infty$ exists, is deterministic and
 can be identified explicitly (Proposition~\ref{propo:longtimelimMI}). In particular, we see
that the cost savings can range from $0\%$ (if $E[\beta_{k}]=1$ and $E[\eta_{k}]>1$) to 
$100\%$ (if $E[\beta_{k}]<1$ and $E[\eta_{k}]\le 1$).

\medskip 
This article is organized as follows.
In Section~\ref{sec:tradeexecutionproblemformulation} we introduce the mathematical setting, state the stochastic control problem and provide its financial interpretation.
In Section~\ref{sec:maintheoremsection} we solve the problem via dynamic programming, study the existence of the long-time limit $\lim_{n\to-\infty}Y_n$ of the characterizing process $Y$ and discuss a few technical issues.
A subsetting where $Y$ becomes deterministic is examined in Section~\ref{sec:PIMI}.
In Section~\ref{sec:roundtrips} we study the existence of profitable round trips and in Section~\ref{sec:1go} we discuss when it is optimal to close the position prematurely; both sections describe several qualitative effects via general statements and examples.
Appendix~\ref{a:proof_mainres} contains the proof of Theorem~\ref{thm:mainres}.
Two simple lemmas on integrability, which we often use in our arguments, are included for convenience in Appendix~\ref{b:integrability}.

\subsection*{Extended literature discussion}

A stream of literature pioneered by \cite{kyle1985continuous} studies the underlying mechanism for the formation of illiquidity. Our paper is part of another branch that takes price impact as exogeneously given. 
In order to embed our paper into related literature on optimal trade execution we divide that branch of the literature into two groups depending on how the price impact is modeled. 
The market models in Group~A  assume that the price impact splits into two components.
There is a temporary (or instantaneous) component that only affects the current trade and a permanent component which affects all future trades equally. The pioneering works
\cite{almgren1999value,almgren2001optimal,bertsimas1998optimal}
assume in a discrete-time framework that both components are proportional to the trade size.
Under this assumption the permanent component has no effect on determining the optimal execution strategy.
Execution strategies in continuous-time models within Group~A have absolutely continuous paths and can therefore be described by their derivatives --- the so-called trading rates (see, e.g., \cite{almgren2003optimal}).
The papers
\cite{schied2009risk,SchSchTeh2010}
include risk aversion into the continuous-time model of Group~A with constant price impact coefficient.
The papers
\cite{bank2017hedging,bank2018linear}
discuss hedging with instantaneous price impact. 
The models of Group~A are extended to nonlinear (mostly power law-shaped) dependencies of the price impact on the trading rates and to random price impact coefficients in, e.g.,
\cite{almgren2003optimal,
almgren2012optimal,
ankirchner2020optimal,
ankirchner2014bsdes,
ankirchner2015optimal,
bank2018linear,
cheridito2014optimal,
dolinsky2020note,
graewe2015non,
graewe2018smooth,
horst2016constrained,
kruse2016minimal,
popier2019second,
schied2013control}.

Motivated by empirical studies of limit order books, the models of Group~B postulate that trades have a transient price impact that decays over time due to resilience effects of the price.
The pioneering works \cite{alfonsi2008constrained,obizhaeva2013optimal} model the price impact via a block-shaped limit order book, where the impact decays exponentially at a constant rate.
Subsequent works within Group~B either extend this framework in different directions or suggest alternative frameworks with similar features.
There is a subgroup of models which include more general limit order book shapes \cite{alfonsi2010optimal,alfonsi2010boptimal,predoiu2011optimal}.
Models in another subgroup extend the exponential decay of the price impact to general decay kernels \cite{alfonsi2012order,gatheral2012transient}.
In the framework with exponential decay kernel \cite{lorenz2013drift} allows a non-martingale dynamics for the unaffected price and discusses the dependence of optimal trade execution strategies on the drift.
Models of Group~B with transient multiplicative price impact have recently been analyzed in \cite{becherer2018optimala,becherer2018optimalb}, whereas \cite{becherer2019stability} contains a stability result for the involved cost functionals. Optimal investment and superreplication in a block-shaped limit order book model with exponential resilience is discussed in \cite{BD_AAP_2019,BD_Bern_2020,BV_SIFIN_2019}.
The present paper and its continuous-time counterpart \cite{ackermann2020cadlag} fall into the subgroup of Group~B that studies time-dependent (possibly stochastic) limit order book depth and resilience \cite{alfonsi2014optimal,bank2014optimal,fruth2014optimal,fruth2019optimal}. The differences between \cite{ackermann2020cadlag,alfonsi2014optimal,bank2014optimal,fruth2014optimal,fruth2019optimal} and our paper were extensively discussed earlier in the introduction, while this literature discussion makes it clear that our paper cannot be assigned to any of the other mentioned subgroups.\footnote{It is worth noting that continuous-time optimization problems within Group~B result in strategies that involve jumps (this allows, in a sense, to profit from resilience). There are, however, several papers that share some features of both Group~A and Group~B.
For instance, \cite{gatheral2010no} analyzes price manipulation in models with general (not time-varying) decay kernels in the class of absolutely continuous strategies.
Optimal control problems in \cite{graewe2017optimal,horst2019multi} arise in the context of trade execution under stochastic resilience (cf.\ Group~B), where the optimization is carried out over absolutely continuous strategies (as in Group~A).
The jumps do not appear due to the form of the cost functionals in \cite{graewe2017optimal,horst2019multi}.}

One might alternatively classify the literature on optimal trade execution depending on the type of mathematics arising. For instance, the papers in Group~A, where the price impact depends linearly on the trading rate, and the papers in Group~B, where the underlying limit order book is block-shaped, lead to linear-quadratic stochastic control problems (which are different for different papers listed above). In particular, this applies to our present paper, but due to the listed essential differences between the settings, the linear-quadratic problems in aforementioned papers cannot imply our results, and a separate analysis, which we perform in the present paper, is required.

\section{A trade execution problem with stochastic market depth and stochastic resilience}
\label{sec:tradeexecutionproblemformulation}

In this section we introduce a financial market model where liquidity varies randomly in time. We first give the comprehensive mathematical formulation of the model and subsequently comment on its financial motivation.

\paragraph{Mathematical formulation}
Let $(\Omega, \mathcal F, (\mathcal F_k)_{k\in \Z}, P)$ be a filtered probability space. 
Denote $\Linfm = \bigcap_{p\in [1,\infty)} L^p(\Omega, \mathcal F, P)$ and $\Ltwop = \bigcup_{\varepsilon>0} L^{2+\varepsilon}(\Omega, \mathcal F, P)$. 
Let $\beta=(\beta_k)_{k\in \Z}$ and $\gamma=(\gamma_k)_{k\in \Z}$ be strictly positive adapted stochastic processes, called the \emph{resilience} and the \emph{price impact} process, respectively. Assume that $\beta_k, \gamma_k  \in \Linfm$ for all $k\in \Z$. Furthermore, it turns out to be convenient to denote the multiplicative increments of $\gamma$ by $\eta_n=\frac{\gamma_n}{\gamma_{n-1}}$, $n\in \Z$.

Let $N\in\N$. 
For $n \in \Z \cap (-\infty,N]$ and $x\in \R$ we call a real-valued adapted stochastic process $\xi=(\xi_k)_{k\in\{n,\ldots, N\}}$ satisfying $x+\sum_{j=n}^N \xi_j=0$ an \emph{execution strategy}. 
We denote by $\mathcal A_n(x)$ the set of all execution strategies $\xi$ with $\xi_k \in \Ltwop$  for all $k \in \{n,\ldots, N\}$. 
For an execution strategy $\xi \in \mathcal A_n(x)$ we call the process $X=(X_k)_{k\in\{n,\ldots, N\}}$ satisfying $X_k=x+\sum_{j=n}^k \xi_j$, $k\in \{n,\ldots, N\}$ the \emph{position path} associated to $\xi$. For $d\in \R$ and $\xi \in \mathcal A_n(x)$ we define the \emph{deviation process}
$D=(D_{k-})_{k \in \{n,\ldots,N\}}$ associated to $\xi$ recursively by
\begin{equation}\label{eq:13052019a1}
D_{n-}=d
\quad\text{and}\quad
D_{k-}=(D_{(k-1)-}+\gamma_{k-1}\xi_{k-1})\beta_k,\quad
k\in\{n+1,\ldots,N\}.
\end{equation}
Note that the process $D=(D_{k-})_{k \in \{n,\ldots,N\}}$ is adapted. 
The value function $V\colon \Omega \times (\Z \cap (-\infty,N]) \times \R \times \R \to \R$ of the control problem is given by
\begin{equation}\label{eq:value_fct}
V_n(x,d)=\essinf_{\xi \in \mathcal A_n(x)} E_n\left[\sum_{j=n}^N \left(D_{j-}+\frac{\gamma_j}{2}\xi_j\right)\xi_j \right], \quad n\in \Z \cap (-\infty,N], x \in \R , d\in \R,
\end{equation}
where the argument $d$ is the starting point of the process $D$
in~\eqref{eq:13052019a1},
and $E_n[\cdot]$ is a shorthand notation for $E[\cdot|\cF_n]$.

\paragraph{Financial interpretation}
The numbers $N\in \N$ and $n \in \Z \cap (-\infty,N]$ specify the end and the beginning of the trading period, respectively. The possible trading times are given by the set $\{n,\ldots, N\}$. The number $x\in \R$ represents the initial position of the agent.
A negative $x<0$ means that the agent has to buy $|x|$ shares over the trading period,
while a positive $x>0$ means that the agent has to sell $x$ shares over the trading period.
For an execution strategy $\xi \in \mathcal A_n(x)$ the value of $\xi_k$ specifies the number of shares traded by the agent at time $k\in \{n,\ldots,N\}$. 
A positive value $\xi_k>0$ means that the agent buys shares, whereas a negative value $\xi_k<0$ corresponds to selling. 
For the associated position path $X$ the value of $X_k$ represents the agent's position at time  $k\in \{n,\ldots,N\}$ \emph{directly after} the trade $\xi_k$. Observe that all position paths satisfy $X_N=0$, i.e., the position is closed after the last trade at time $N$. The process $D$ describes the deviation of the price of a share from the unaffected
price caused by the past trades of the agent 
(see also Remark~\ref{rem:unaffected_price_process_martingale} below where we explain that explicitly including an unaffected price process modelled by a square integrable martingale does not change the control problem~\eqref{eq:value_fct}). 
Given a deviation of size $D_{(k-1)-}$ \emph{directly prior} to the trade at time $k-1$, the deviation directly after a trade of size $\xi_{k-1}$ equals $D_{(k-1)-}+\gamma_{k-1}\xi_{k-1}$. In particular, the change of the deviation is proportional to the size of the trade and the proportionality factor is given by the price impact process $\gamma$. In the language of the literature on optimal trade execution problems our model thus includes a linear price impact. This corresponds to a block-shaped symmetric limit order book, i.e., limit orders are uniformly distributed to the left and to the right of the mid-market price. 
Note that in our idealized model the bid-ask spread is always assumed to be $0$.  
The height of the order book at time $k$ is given by $1/\gamma_k$. In particular, our model allows the height of the limit order book to evolve randomly in time and thereby captures stochastic market liquidity.
Note that since $\gamma$ is positive, a purchase $\xi_k>0$ increases the deviation whereas a sale $\xi_k<0$ decreases it.
In the period after the trade at time $k-1$ and before the trade at time $k$ the deviation changes from $D_{(k-1)-}+\gamma_{k-1}\xi_{k-1}$ to $D_{k-}=(D_{(k-1)-}+\gamma_{k-1}\xi_{k-1})\beta_k$ 
due to resilience effects in the market. 
In the literature on optimal execution the resilience process
$\beta$ is often assumed to take values in $(0,1)$
and describes the speed with which the deviation tends back to zero between two trades, 
where values of $\beta$ close to zero signify a faster reversion to zero. 
In this case, i.e., for $(0,1)$-valued $\beta$, the price impact is usually called transient (cf., e.g., \cite{alfonsi2012order}). 
The case $\beta\equiv 1$ corresponds to permanent impact. 
In our work we assume $\beta$ only to be positive.
A value $\beta_k>1$ describes the effect when the deviation continues to move in the direction of the trade for some time after the trade. 
Note that also $\beta$ evolves randomly in time. In particular, when making a decision about the size of the trade at time $k-1$, the agent, in general, cannot predict the exact impact of this trade on the future price at time $k$ 
(as $\beta_k$ is only $\cF_k$-measurable). 
At each time $k\in \{n,\ldots, N\}$ the liquidity costs  incurred by a trade $\xi_k$ amount to $(D_{k-}+\frac{\gamma_k}{2}\xi_k)\xi_k$. 
This means that the price per share that the agent has to pay  
in addition to the unaffected price equals the mean of the deviation before the trade $D_{k-}$
and the deviation after the trade $D_{k-}+\gamma_{k}\xi_{k}$. 
Control problem~\eqref{eq:value_fct} thus corresponds to minimizing the expected costs of closing an initial position of size $x$ within the trading period $\{n,\ldots, N\}$ given initial deviation $d$,
where the minimization is performed in an extension of block-shaped limit order book models to the case of randomly evolving order book depth and resilience.
For a more detailed description of the ideas behind limit order book models 
we refer to any of the papers
\cite{alfonsi2008constrained,alfonsi2010optimal,alfonsi2010boptimal,obizhaeva2013optimal,predoiu2011optimal}.

\medskip
We conclude this section with some remarks on the well-posedness of the optimal trade execution problem \eqref{eq:value_fct} and a possible extension of the model.

\begin{remark}\label{rem:DinL2+}
Let $n\in \Z \cap (-\infty,N]$, $x, d \in \R$ and $\xi \in \mathcal A_n(x)$. Then for the associated deviation process $(D_{k-})_{k \in \{n,\ldots,N\}}$ it holds that  $D_{k-} \in \Ltwop$ for all $k \in \{n,\ldots, N\}$. 

We prove this claim by induction on $k$. Since $D_{n-}=d$, the claim obviously holds true for $k=n$. Consider the step $\{n,\ldots,N-1\} \ni k-1 \to k \in \{n+1,\ldots,N\}$ and note that by the Minkowski inequality and \eqref{eq:13052019a1}, it is sufficient to show that $D_{(k-1)-}\beta_k \in \Ltwop$ and $\gamma_{k-1}\xi_{k-1}\beta_k  \in \Ltwop$. 
Since $\beta_k \in \Linfm$ and, by the induction hypothesis, $D_{(k-1)-} \in \Ltwop$, Lemma~\ref{lem:inL2+} proves that 
$D_{(k-1)-}\beta_k \in \Ltwop$. For $\gamma_{k-1}\xi_{k-1}\beta_k$ observe first that $\gamma_{k-1}\beta_k \in \Linfm$ since both factors belong to $\Linfm$. Then, recall that $\xi_{k-1} \in \Ltwop$ and apply
Lemma~\ref{lem:inL2+} to obtain that $\gamma_{k-1}\xi_{k-1}\beta_k \in \Ltwop$.
\end{remark}

\begin{remark}\label{rem:valuefctwelldefined}
Note that the value function is well-defined. To show this, we verify that for all $n\in \Z \cap (-\infty,N]$, $x, d \in \R$, $\xi \in \mathcal A_n(x)$ each summand $\left(D_{j-}+\frac{\gamma_j}{2}\xi_j\right)\xi_j$, $j\in\{n,\ldots,N\}$, is integrable.

Since $\gamma_j \in \Linfm$ and $\xi_j\in \Ltwop$, it follows from Lemma~\ref{lem:inL2+} that the product $\gamma_j\xi_j$ is in $\Ltwop$. By Remark \ref{rem:DinL2+}, $D_{j-} \in \Ltwop$ as well. Hence, $D_{j-}$ and $\gamma_j \xi_j$ are square integrable and so is $D_{j-} + \frac{\gamma_j}{2}\xi_j$. Furthermore, $\xi_j$ is square integrable as it is in $\Ltwop$. The Cauchy-Schwarz inequality thus yields the integrability of $\left(D_{j-}+\frac{\gamma_j}{2}\xi_j\right)\xi_j$.
\end{remark}

\begin{remark}\label{rem:formD}
For $n\in \Z \cap (-\infty,N]$, $x, d \in \R$ and $\xi \in \mathcal A_n(x)$ the deviation process $D=(D_{k-})_{k \in \{n,\ldots,N\}}$ associated to $\xi$ is given explicitly by 
\begin{equation}\label{eq:formD}
D_{k-} = d \prod_{l=n+1}^{k} \beta_{l} + \sum_{i=n+1}^{k} \gamma_{i-1} \xi_{i-1} \prod_{l=i}^{k} \beta_{l}, \quad k \in \{n,\ldots,N\}.
\end{equation}
This can be established by induction on $k \in \{n,\ldots,N\}$.
\end{remark}

\begin{remark}\label{rem:unaffected_price_process_martingale}
One can also include an unaffected price process in the model.
Indeed, if the unaffected price process is given by the square integrable martingale $S=\left( S_k \right)_{k \in \Z\cap(-\infty,N]}$, then, for all $n\in \Z \cap (-\infty,N]$, $x\in \R$ and $\xi\in\mathcal{A}_n(x)$,
with the notation $X_{n-1}=x$, we get
$$
E_n\left[ \sum_{j=n}^N S_j \xi_j \right] 
=
E_n\left[ \sum_{j=n}^N S_j (X_j-X_{j-1}) \right]
=
E_n\left[-xS_n-\sum_{j=n}^{N-1} X_j(S_{j+1}-S_j) \right]  
=-xS_n.
$$
It follows that for all $n\in \Z \cap (-\infty,N]$ and $x, d \in \R$ the expected costs generated by an execution strategy $\xi \in \mathcal A_n(x)$ with the deviation process $(D_{k-})_{k \in \{n,\ldots,N\}}$ of~\eqref{eq:13052019a1} satisfy
\begin{equation}\label{eq:expectationinvaluefunctionforunaff}
E_n\left[ \sum_{j=n}^N \left(S_j + D_{j-}+\frac{\gamma_j}{2}\xi_j\right)\xi_j \right] = -x S_n + E_n\left[ \sum_{j=n}^N \left( D_{j-}+\frac{\gamma_j}{2}\xi_j\right)\xi_j \right] .
\end{equation}
Hence, minimizing $E_n\left[ \sum_{j=n}^N \left(S_j + D_{j-}+\frac{\gamma_j}{2}\xi_j\right)\xi_j \right] $ is equivalent
to~\eqref{eq:value_fct}.

The literature on optimal trade execution is primarily concerned with the minimization of implementation costs caused by limited market liquidity
 rather than with the generation of trading gains by exploiting trends in the underlying price process. Therefore, the assumption that the unaffected price process is a martingale is fairly standard in the literature and has already been made in many previous papers (e.g., \cite{alfonsi2008constrained,alfonsi2010optimal,alfonsi2010boptimal,obizhaeva2013optimal,predoiu2011optimal}). There are, however, also some papers that analyze the dependence of optimal trade execution strategies on a possible drift in the underlying unaffected price process (see, e.g., \cite{ankirchner2015optimal,lorenz2013drift}).
\end{remark}

\section{Characterization of minimal costs and optimal strategies}
\label{sec:maintheoremsection}
The following result provides a solution to the stochastic control problem~\eqref{eq:value_fct}.
It shows that the value function and the optimal strategy in~\eqref{eq:value_fct} are characterized by a single process $Y$ that is defined via a backward recursion.

\begin{theo}\label{thm:mainres}
Assume that for all $n\in \Z\cap(-\infty,N]$ we have $\beta_n, \gamma_n, \frac{1}{\gamma_n} \in \Linfm$ and that for all $n\in \Z \cap (-\infty,N-1]$ it holds that $E_n\left[\frac{\beta^2_{n+1}}{\eta_{n+1}}\right]<1$ a.s.\ and,
with $\alpha_n=1-E_n\left[\frac{\beta^2_{n+1}}{\eta_{n+1}}\right]$,
we have $\frac1{\alpha_n}\in \Linfm$. 
Let $(Y_n)_{n \in \Z \cap (-\infty,N]}$ be the process that is recursively defined by $Y_N=\frac{1}{2}$ 
and
\begin{equation}\label{eq:defY}
Y_n=E_n[\eta_{n+1}Y_{n+1}]-\frac{\left(E_n\left[Y_{n+1}\left(\beta_{n+1}-\eta_{n+1}\right) \right]\right)^2}{E_n\left[\frac{Y_{n+1}}{\eta_{n+1}}\left(\beta_{n+1}-\eta_{n+1}\right)^2 +\frac{1}{2}\left(1-\frac{\beta_{n+1}^2}{\eta_{n+1}}\right)\right]}, n\in \Z \cap (-\infty,N-1].
\end{equation}
Then it holds for all $n \in \Z \cap (-\infty,N]$, $x,d \in \R$ that
\begin{equation}\label{eq:sol_valfct}
V_n(x,d)=\frac{Y_n}{\gamma_n}\left(d-\gamma_n x\right)^2-\frac{d^2}{2\gamma_n} \quad \text{ and } \quad 0<Y_n\leq\frac{1}{2}.
\end{equation}
Moreover, for all $x,d \in \R$ the (up to a $P$-null set) unique
optimal trade size is given by
\begin{equation}\label{eq:optstrat}
\xi^*_n(x,d)=\frac{E_n\left[Y_{n+1}\left(\beta_{n+1}-\eta_{n+1}\right) \right]}{E_n\left[\frac{Y_{n+1}}{\eta_{n+1}}\left(\beta_{n+1}-\eta_{n+1}\right)^2+\frac{1}{2}\left(1-\frac{\beta_{n+1}^2}{\eta_{n+1}}\right) \right]}\left(x-\frac d{\gamma_n}\right)-\frac{d}{\gamma_n}, n \in \Z \cap (-\infty,N-1],
\end{equation}
and $\xi_N^*(x,d)=-x$,
and we have $\xi_n^*(x,d)\in\Linfm$
for all $n\in\bbZ\cap(-\infty,N]$ and $x,d\in\bbR$. 

In particular, for all $n\in \Z \cap (-\infty,N]$, $x,d \in \R$ the process $\xi^*=\left(\xi_k^*\right)_{k \in \{n,\ldots,N\}}$ recursively defined by $X_{n-1}^*=x, D_{n-}^*=d$, 
\begin{equation}\label{eq:optprocess}
\begin{split}
\xi_k^*=\xi_k^*\left( X_{k-1}^*, D_{k-}^* \right), \, X_k^*=X_{k-1}^*+\xi_k^*, \, D_{(k+1)-}^*=\left( D_{k-}^* +\gamma_k \xi_k^*\right) \beta_{k+1}, \, k \in \{n,\ldots,N\}
\end{split}
\end{equation}
is a unique optimal strategy in $\mathcal A_n(x)$ for~\eqref{eq:value_fct}.
\end{theo}

The proof of Theorem~\ref{thm:mainres} is deferred to Appendix~\ref{a:proof_mainres}. 
We now present an example of a reasonably large class of models where the assumptions of Theorem~\ref{thm:mainres} are automatically satisfied, and we discuss tractability of \eqref{eq:defY} and~\eqref{eq:optstrat} in Markovian situations.

\begin{ex}
(i) Let $(a_n)_{n\in\bbZ\cap(-\infty,N]}$ and $(\beta_n)_{n\in\bbZ\cap(-\infty,N]}$ be deterministic strictly positive sequences such that
$$
\frac{\beta_{n+1}^2a_n}{a_{n+1}}<1\quad\text{for all }n\in\bbZ\cap(-\infty,N-1]
$$
(a particular case: $\beta_n\equiv\beta\in(0,1)$ and $(a_n)$ nondecreasing).
Let $(\gamma_n)_{n\in\bbZ\cap(-\infty,N]}$ be given by the formula $\gamma_n=\frac{a_n}{Z_n}$, where $(Z_n)_{n\in\bbZ\cap(-\infty,N]}$ is a strictly positive supermartingale such that $Z_n,\frac1{Z_n}\in L^{\infty-}$ for all $n\in\bbZ\cap(-\infty,N]$. It is straightforward to see that all assumptions of Theorem~\ref{thm:mainres} are satisfied.

\smallskip
(ii) Let $(\beta_n)_{n\in\bbZ\cap(-\infty,N]}$ and $(\gamma_n)_{n\in\bbZ\cap(-\infty,N]}$ satisfy the assumptions of Theorem~\ref{thm:mainres} (a particular case: part~(i) in this example) and possess the structure
$$
\beta_n=f_n(\theta_n)\quad\text{and}\quad\gamma_n=g_n(\theta_n)\quad\text{for all }n\in\bbZ\cap(-\infty,N],
$$
where $f_n$ and $g_n$ are measurable functions and $\theta=(\theta_n)$ is an $(\cF_n)$-Markov process. Then we get from~\eqref{eq:defY} and~\eqref{eq:optstrat} that
$$
Y_n=h_n(\theta_n)\quad\text{and}\quad\xi^*_n(x,d)=k_n(\theta_n)\left(x-\frac d{\gamma_n}\right)-\frac d{\gamma_n}
$$
for all $n\in\bbZ\cap(-\infty,N-1]$ and $x,d\in\bbR$, where the functions $h_n$ and $k_n$ are computed via the integration with respect to the transition kernels of the Markov process $\theta$.
While, in general, this can be a challenging task, which needs to be carried out numerically, in Section~\ref{sec:PIMI} we consider a specific situation which can be treated fully explicitly.
\end{ex}

We can give the following interpretation to the process $Y$ from
Theorem~\ref{thm:mainres}:
Suppose that at time $n\in \Z\cap (-\infty,N]$ the task is to sell $x=1$ share given an initial deviation of $d=0$. Then immediate execution of the share generates the costs $\frac{\gamma_n}2$.
The optimal execution strategy incurs the expected costs
$V_n(1,0)=\gamma_nY_n$ (recall~\eqref{eq:sol_valfct}).
So, the random variable $2Y_n\colon \Omega \to [0,1]$ describes 
to which percentage the costs of selling the unit immediately can be reduced by executing the position optimally.

Next we analyze how much better in comparison to the immediate closure we can do in the long run. 
There are basically two starting points. 
One is to adopt the perspective that 
trading starts at a fixed point in time, e.g., at $n=0$, and that the terminal date $N$ when the position has to be closed is shifted further and further into the future. 
This corresponds to studying the limit of the sequence of random variables $(Y_0^N)_{N \in \N}$ as $N\to \infty$, where $Y^N$ is the process defined as in \eqref{eq:defY} pertaining to the terminal time $N$.\footnote{Note that the filtered probability space $(\Omega, \mathcal F, (\mathcal F_k)_{k\in \Z}, P)$ and the processes $(\gamma_k)_{k\in \Z}$, $(\beta_k)_{k\in \Z}$ do not depend on $N$.}
Observe that the fact that $Y_0^N=V_0^N(1,0)/\gamma_0$ is nonnegative and nonincreasing in $N$, where $V^N$ is the value function belonging to the terminal time $N$, implies that $\lim_{N\to \infty} Y_0^N$ always exists (under the assumptions of Theorem~\ref{thm:mainres}).
Another perspective consists in fixing the terminal time $N$ and asking what would have been if one had started trading earlier. This corresponds to investigating $\lim_{n\to -\infty} Y_n$. 
In some settings
(e.g., in a time-homogeneous deterministic framework or, more generally, in the setting of Proposition~\ref{propo:longtimelimMI} below) one can see that both perspectives coincide by simply relabeling time instances appropriately.
In the next proposition we study the existence of the long-time limit $\lim_{n\to-\infty}Y_n$. 
In contrast to $\lim_{N\to \infty} Y_0^N$ this limit does not always exist (cf.\ Lemma~\ref{lemma:Ydoesnotconverge}).
We further refer to Proposition~\ref{propo:longtimelimMI} and the extensive discussion after it, where a specific framework, in which both perspectives coincide, is treated in more detail.\footnote{We remark that the question of the long-time limit is different from considering the continuous-time limit of the control problem, which corresponds to fixing $N\in \N$ and $n\in \Z \cap (-\infty,N]$ and letting the number of available trading times in $[n,N]$ go to infinity. In \cite{ackermann2020cadlag} we use some of the results of the present paper to derive a quadratic BSDE which describes the continuous-time limit of the process~$Y$.}

\begin{propo}\label{prop:08062019a1}
Let the assumptions of Theorem~\ref{thm:mainres} be in force.
Fix any $p\in[1,\infty)$.

\smallskip
(i) The sequence $(\gamma_nY_n)_{n\in\bbZ\cap(-\infty,N]}$
converges a.s.\ and in $L^p$
as $n\to-\infty$ to a finite nonnegative random variable.

\smallskip
(ii) If $(\gamma_n)_{n\in\bbZ\cap(-\infty,N]}$
is a supermartingale, then the sequence
$(Y_n)_{n\in\bbZ\cap(-\infty,N]}$
converges a.s.\ and in $L^p$
as $n\to-\infty$ to a finite nonnegative random variable.
\end{propo}

The assumption that
$(\gamma_n)_{n\in\bbZ\cap(-\infty,N]}$
is a supermartingale in~(ii)
means that the liquidity in the model
increases in time (in average).
In Lemma~\ref{lemma:Ydoesnotconverge} below 
$(\gamma_n)_{n\in\bbZ\cap(-\infty,N]}$
is a submartingale and
$(Y_n)_{n\in\bbZ\cap(-\infty,N]}$
does not converge.
This shows that the claim in~(ii)
does not in general hold in the situation
when the liquidity in the model decreases in time.

\begin{proof}[Proof of Proposition~\ref{prop:08062019a1}]
(i) It follows from~\eqref{eq:defY}
that for all $n\in\bbZ\cap(-\infty,N-1]$ it holds
$Y_n\le E_n[\eta_{n+1}Y_{n+1}]=\frac{1}{\gamma_n}E_n[\gamma_{n+1}Y_{n+1}]$. Thus, $(\gamma_nY_n)_{n\in\bbZ\cap(-\infty,N]}$
is a submartingale. 
Hence it converges a.s.\ as $n\to-\infty$
due to the backward convergence theorem.
Moreover, $(\gamma_nY_n)_{n\in\bbZ\cap(-\infty,N]}$
is a positive sequence in $L^{\infty-}$,
and, by the Jensen inequality, $(\gamma_nY_n)^p\le E_n[(\gamma_NY_N)^p]$, $n\in\bbZ\cap(-\infty,N]$,
hence the sequence
$((\gamma_nY_n)^p)_{n\in\bbZ\cap(-\infty,N]}$
is uniformly integrable.
This implies the convergence in $L^p$.

\smallskip
(ii) If $(\gamma_n)_{n\in\bbZ\cap(-\infty,N]}$
is a supermartingale,
then it converges a.s.\ as $n\to-\infty$
to a $\bbR\cup\{+\infty\}$-valued random variable,
denoted by $\gamma_{-\infty}$,
due to the backward convergence theorem.
As the process $(\gamma_n)$ is positive,
$\gamma_{-\infty}$ is, in fact, $[0,+\infty]$-valued.
Furthermore, it holds\footnote{Here we use the convention $\infty \cdot 0=0$.}
$$
0=
E\left[\gamma_{-\infty}1_{\{\gamma_{-\infty}=0\}}\right]
\ge
E\left[\gamma_N1_{\{\gamma_{-\infty}=0\}}\right]
\ge0.
$$
Together with the fact that $\gamma_N>0$ a.s.,
this implies $\gamma_{-\infty}>0$ a.s.
It now follows from~(i) that
$(Y_n)_{n\in\bbZ\cap(-\infty,N]}$
converges a.s.\ as $n\to\infty$.
As the sequence $(Y_n)_{n\in\bbZ\cap(-\infty,N]}$
is bounded (being $(0,\frac12]$-valued),
it also converges in $L^p$.
\end{proof}

Now the question arises of whether we can compute the long-time limit $\lim_{n\to-\infty}Y_n$ in specific examples.
In Proposition~\ref{propo:longtimelimMI} below we compute this limit in the framework where $\eta_{k+1}$ and $\beta_{k+1}$ are independent of $\mathcal F_k$ for all $k\in \Z$.
We next present several examples that fall outside this framework.

\begin{ex} (1)
A simple observation is that, under the assumptions of Theorem~\ref{thm:mainres}, we have $\lim_{n\to-\infty}Y_n=0$ a.s.\ whenever $(\gamma_n)$ satisfies $\lim_{n\to-\infty}\gamma_n=+\infty$ a.s.
This follows from statement~(i) of Proposition~\ref{prop:08062019a1}.

\smallskip (2)
Consider the setting where, in addition to the assumptions of Theorem~\ref{thm:mainres}, we have
\begin{equation}\label{eq:11012021a1}
\eta_n=\beta_n\quad\text{for all }n\in\bbZ\cap(-\infty,N].
\end{equation}
It is worth noting that, in this setting, the optimal strategy is to wait until the terminal time $N$ and to close the position at time $N$ whenever the initial deviation $d=0$; while, if $d\ne0$, the optimal strategy in general consists of non-trivial trades at all time points.
For the sake of discussing the long-time limit $\lim_{n\to-\infty}Y_n$ in this setting we observe that
\begin{equation}\label{eq:11012021a2}
E_n[\eta_{n+1}]<1\;\;\text{a.s.\ for all }n\in\bbZ\cap(-\infty,N-1]
\end{equation}
(which is just the requirement $E_n\left[\frac{\beta_{n+1}^2}{\eta_{n+1}}\right]<1$ a.s.\ under~\eqref{eq:11012021a1}).
Hence $(\gamma_n)_{n\in\bbZ\cap(-\infty,N]}$ is a supermartingale.
By statement~(ii) of Proposition~\ref{prop:08062019a1},
$\lim_{n\to-\infty}Y_n$ always exists in this setting.
Moreover, we have $Y_n=E_n[\eta_{n+1}Y_{n+1}]$ for all $n \in \bbZ\cap(-\infty,N-1]$, and hence by induction
\begin{equation}\label{eq:11012021a3}
Y_n=\frac{1}{2}E_n\left[\, \prod_{j=n+1}^N \eta_j \right]
\quad\left(=\frac12\,\frac{E_n[\gamma_N]}{\gamma_n}\right)
\end{equation}
for all $n\in\bbZ\cap(-\infty,N-1]$.
In general, we still can have different values for the long-time limit.
Therefore, we now discuss several more specific examples.

\smallskip (2a)
Assume there exists $c \in (0,1)$ such that $E_n[\eta_{n+1}]\le c$ a.s.\ (cf.~\eqref{eq:11012021a2}) for all $n \in \bbZ\cap(-\infty,N-1]$.
By intermediate conditioning, it follows from~\eqref{eq:11012021a3} that 
$Y_n\le\frac12 c^{N-n}$ a.s. for all $n \in \bbZ\cap(-\infty,N-1]$, hence $\lim_{n\to-\infty}Y_n = 0$ a.s.

\smallskip (2b)
On the other hand, it is clear from~\eqref{eq:11012021a3} that, even with suitable deterministic sequences $(\eta_n)$, we can achieve for the long-time limit $\lim_{n\to-\infty}Y_n$ any deterministic value in $(0,\frac12)$.

\smallskip (2c)
In order to present an explicit and, possibly, non-deterministic long-time limit, we finally consider the following construction.
Let $(a_n)_{n\in \Z\cap (-\infty,N]}\subseteq [0,\infty)$ be a strictly decreasing sequence of nonnegative real numbers,
$Z_i$, $i\in \Z\cap (-\infty,N]$, and $\zeta$ random variables such that $(Z_i)_{i\in \Z\cap (-\infty,N]}$ is an i.i.d.\ sequence independent of $\zeta$,
and $Z_N,\zeta\ge0$, $Z_N,\zeta\in L^{\infty-}$.
We also require at least one of the conditions (a) $a_N>0$ or (b) $\frac1{Z_N},\frac1\zeta\in L^{\infty-}$.
We now define
$$
S_n=\sum_{i=n}^N Z_i,\quad
\cF_n=\sigma(\zeta,S_i;i\in\bbZ\cap(-\infty,n]),\quad
n\in\bbZ\cap(-\infty,N],
$$
and set $\gamma_n=a_n+\frac1{N-n+1}S_n\zeta$ and $\beta_n=\eta_n=\frac{\gamma_n}{\gamma_{n-1}}$, $n\in \Z\cap (-\infty,N]$.
Thus, we are in setting~\eqref{eq:11012021a1}, and we now verify that the assumptions of Theorem~\ref{thm:mainres} are satisfied.
Indeed, the requirement $\beta_n,\gamma_n,\frac1{\gamma_n}\in L^{\infty-}$, $n\in\bbZ\cap(-\infty,N]$, is clear from the construction
(the condition ``(a) or (b)'' above ensures $\frac1{\gamma_N}\in L^{\infty-}$).
Further, as it is well-known, for $n\in\bbZ\cap(-\infty,N-1]$ and $i\in\bbZ\cap[n,N]$, we have $E_n[Z_i]=\frac1{N-n+1}S_n$, hence
$$
E_n[\gamma_{n+1}]=a_{n+1}+\frac1{N-n+1}S_n\zeta
<a_{n}+\frac1{N-n+1}S_n\zeta=\gamma_n\;\;\text{a.s.},
$$
i.e., requirement~\eqref{eq:11012021a2} holds true
(which is $E_n\left[\frac{\beta_{n+1}^2}{\eta_{n+1}}\right]<1$ a.s.).
Finally, in this setting,
$\left(1-E_n\left[\frac{\beta_{n+1}^2}{\eta_{n+1}}\right]\right)^{-1}=
\frac{\gamma_n}{a_n-a_{n+1}}\in L^{\infty-}$, $n\in\bbZ\cap(-\infty,N-1]$.

By the strong law of large numbers it holds that  $\frac1{N-n+1}S_n\to E[Z_N]$ a.s., as $n\to-\infty$. Setting $a_{-\infty}=\lim_{n\to-\infty}a_n$ ($\in(0,\infty]$), we obtain
$\lim_{n\to-\infty}\gamma_n=a_{-\infty}+E[Z_N]\zeta$ a.s.,
$\lim_{n\to-\infty}E_n[\gamma_N]=a_{N}+E[Z_N]\zeta$ a.s., hence
$$
\lim_{n\to-\infty}Y_n=\frac12\lim_{n\to-\infty}\frac{E_n[\gamma_N]}{\gamma_n}=\frac12\,\frac{a_N+E[Z_N]\zeta}{a_{-\infty}+E[Z_N]\zeta}\;\;\text{a.s.},
$$
which is, in general, non-deterministic.
\end{ex}

The next remark provides an improved upper bound for~$Y$.

\begin{remark}[Upper bound for $Y$]\label{rem:ubY}
Under the assumptions of Theorem~\ref{thm:mainres},
for all $ n \in \Z \cap (-\infty,N]$ it holds that $\gamma_n Y_n=V_n(1,0)$. For an initial position of size $1$ at time $ n \in \Z \cap (-\infty,N]$ a possible execution strategy is to sell the whole unit at a point in time $k\in \{n,\ldots,N\}$. If there is no initial deviation, i.e., $d=0$, it follows that the expected costs of such a strategy amount to $E_n\left[\frac{\gamma_k}{2}\right]$. This implies that $Y_n\le \frac{\min_{k\in \{n,\ldots,N\}}E_n\left[\gamma_k\right]}{2\gamma_n}$, which improves the bound $Y_n\leq \frac{1}{2}$ provided by Theorem~\ref{thm:mainres}.
\end{remark}

Besides some integrability assumptions, Theorem~\ref{thm:mainres} requires that $E_{n}\left[\frac{\beta^2_{n+1}}{\eta_{n+1}}\right]< 1$ a.s.\ for all
$n\in \Z \cap (-\infty,N-1]$. The next remark discusses this assumption.

\begin{remark}[Discussion of the structural assumption]
The assumption
$E_{n}\left[\frac{\beta^2_{n+1}}{\eta_{n+1}}\right]< 1$ a.s.\ for all
$n\in \Z \cap (-\infty,N-1]$ in Theorem~\ref{thm:mainres}
is a certain structural assumption which ensures that 
minimization problem~\eqref{eq:value_fct} is strictly convex.
More precisely, under this assumption
the coefficients $a_n$ in front of $\xi^2$ in~\eqref{eq:gen_dpp}
(see Appendix~\ref{a:proof_mainres})
and the random variables $Y_n$ in~\eqref{eq:defY}
stay positive at all times.
In this remark we show that, on the one hand,
this assumption is in general \emph{not} necessary for that,
but, on the other hand, it guarantees that the problem
preserves the structure with increasing number of time steps.
To this end we consider a two-period version of the problem and distinguish several cases.
First, we recall that $Y_N=\frac12$ and
observe that with \eqref{eq:gen_dpp} it holds for all $x,d \in \R$
\begin{equation}\label{eq:value_fct_two_period}
\begin{split}
V_{N-1}(x,d)&=\essinf_{\xi \in \mathcal S_{N-1}} \bigg\{E_{N-1}[\eta_N+1-2\beta_N]\frac{\gamma_{N-1}\xi^2}{2} 
+E_{N-1}\left[(1-\beta_{N})d-\left(\frac{\beta_{N}}{\eta_N}-1\right)\gamma_Nx \right]\xi \\
&\qquad +E_{N-1}\left[\frac{\gamma_Nx^2}{2}-\beta_Ndx \right]\bigg\},
\end{split}
\end{equation}
where $\mathcal S_{N-1}$ denotes the set of all $\cF_{N-1}$-measurable random variables $\xi\in L^{2+}$.
Next, observe that the process $Y$ defined by \eqref{eq:defY} is given at time $N-1$ by
\begin{equation}\label{eq:Ytwoperiod}
\begin{split}
Y_{N-1}&=E_{N-1}\left[\frac{\eta_N}{2}\right]-\frac{\left(E_{N-1}\left[\beta_{N}-\eta_N \right]\right)^2}{2E_{N-1}[\eta_N-2\beta_N+1]}=\frac{E_{N-1}[\eta_{N}]-(E_{N-1}[\beta_N])^2}{2E_{N-1}[\eta_N-2\beta_N+1]}.
\end{split}
\end{equation}
Moreover, the Cauchy-Schwarz inequality ensures that $(E_{N-1}[\beta_N])^2\le E_{N-1}\left[ \frac{\beta_N^2}{\eta_N}\right]E_{N-1}[\eta_N]$ and hence it holds that
\begin{equation}\label{eq:27052019a1}
\frac{2E_{N-1}[\beta_N]-1}{E_{N-1}[\eta_N]}\le \frac{(E_{N-1}[\beta_N])^2}{E_{N-1}[\eta_N]} \le E_{N-1}\left[ \frac{\beta_N^2}{\eta_N}\right]. 
\end{equation}
In particular, we get the following statements.

\smallskip (i)
On the event $\left\{\frac{2E_{N-1}[\beta_N]-1}{E_{N-1}[\eta_N]}>1\right\}$
the minimization problem in~\eqref{eq:value_fct_two_period} is ill-posed in the sense that it is strictly concave and one can generate infinite gains
(in the limit)
by choosing strategies with $|\xi|\to\infty$.

\smallskip (ii)
On the event $\left\{\frac{2E_{N-1}[\beta_N]-1}{E_{N-1}[\eta_N]}<1<\frac{(E_{N-1}[\beta_N])^2}{E_{N-1}[\eta_N]}\right\}$
there exists a minimizer in~\eqref{eq:value_fct_two_period}.
The random variable $Y_{N-1}$ is, however, negative. As a consequence, in view of~\eqref{eq:abc}, one needs to impose further conditions on $\beta_{N-1}$ and $\eta_{N-1}$ to ensure that the coefficient $a_{N-2}$ is positive and that the minimization problem at time $N-2$ is well-posed.

\smallskip (iii)
On the event
$\left\{\frac{(E_{N-1}[\beta_N])^2}{E_{N-1}[\eta_N]}<1\right\}$,
which is bigger than
$\left\{E_{N-1}\left[\frac{\beta_N^2}{\eta_N}\right]<1\right\}$
(see~\eqref{eq:27052019a1}),
there exists a minimizer in~\eqref{eq:value_fct_two_period} and,
moreover, $Y_{N-1}\in(0,\frac12]$ (see~\eqref{eq:Ytwoperiod}).

\smallskip
Observe, however, that replacing the assumption
$E_{n}\left[\frac{\beta^2_{n+1}}{\eta_{n+1}}\right]< 1$ a.s.\ with
the weaker one
$\frac{(E_n[\beta_{n+1}])^2}{E_n[\eta_{n+1}]}<1$ a.s.\ for all
$n\in\bbZ\cap(-\infty,N-1]$
does not in general allow to perform the backward induction,
as the structure of the problem can be lost
already on the step $N-1\to N-2$.
Namely, $Y_{N-1}$ can be strictly less than $\frac12$
(in contrast to $Y_N=\frac12$),
while
$E_{N-2}\left[\frac{\beta_{N-1}^2}{\eta_{N-1}}\right]$
can be strictly bigger than $1$
(even assuming
$\frac{(E_{N-2}[\beta_{N-1}])^2}{E_{N-2}[\eta_{N-1}]}<1$ a.s.),
and we do not necessarily get positivity of $a_{N-2}$
(see~\eqref{eq:abc}).

\end{remark}

The next remark reveals the following property of optimal strategies: 
Irrespectively of the position $x$ and the deviaton $d$ prior to the trade at time $n$, the ratio between position and deviation after the trade $\xi^*_n(x,d)$ is given by an $\mathcal F_n$-measurable random variable $z_n$ (that does not depend on $(x,d)$).

\begin{remark}[Optimal deviation-position ratio]
In the setting of Theorem~\ref{thm:mainres} the optimal position path can be characterized in terms of its ratio to the associated deviation process. More precisely, let $z=(z_n)_{n \in \Z \cap (-\infty,N]}$ be the
$\bbR\cup\{\infty\}$-valued
adapted process given by
\begin{equation}
z_n=\frac{\gamma_n E_n\left[Y_{n+1}\left(\beta_{n+1}-\eta_{n+1}\right) \right]}
{E_n\left[
\left(Y_{n+1}-\frac12\right)\frac{\beta_{n+1}^2}{\eta_{n+1}}
-Y_{n+1}\beta_{n+1}+\frac12
\right]}, \quad n \in \Z \cap (-\infty,N-1],
\quad z_N=\infty,
\end{equation}
where we set $\frac a0=\infty$ whenever $a\in\bbR\setminus\{0\}$.
Notice that the fraction defining $z_n$, $n\in\Z\cap(-\infty,N-1]$, a.s.\ does not produce $\frac00$ because
\begin{equation*}
\begin{split}
&E_n\left[
\left(Y_{n+1}-\frac12\right)\frac{\beta_{n+1}^2}{\eta_{n+1}}
-Y_{n+1}\beta_{n+1}+\frac12
\right]
-E_n\left[Y_{n+1}\left(\beta_{n+1}-\eta_{n+1}\right) \right]\\
&=E_n\left[
\frac12\left(
1-\frac{\beta_{n+1}^2}{\eta_{n+1}}
\right)
+\frac{Y_{n+1}}{\eta_{n+1}}
(\beta_{n+1}-\eta_{n+1})^2
\right]>0\;\;\text{a.s.}
\end{split}
\end{equation*}
under the assumptions of Theorem~\ref{thm:mainres}.
Then for all $n \in \Z \cap (-\infty,N-1]$, $x,d\in \R$,
$d\ne\gamma_n x$,
the ratio between the deviation $d+\gamma_n\xi^*_n(x,d)$ and the position $x+\xi^*_n(x,d)$ directly after the optimal trade equals
\begin{equation}
\begin{split}
\frac{d+\gamma_n\xi^*_n(x,d)}{x+\xi^*_n(x,d)}
&=\frac{\gamma_nE_n\left[Y_{n+1}\left(\beta_{n+1}-\eta_{n+1}\right) \right]}{E_n\left[Y_{n+1}\left(\beta_{n+1}-\eta_{n+1}\right) \right]+ E_n\left[\frac{1}{2}\left(1-\frac{\beta_{n+1}^2}{\eta_{n+1}}\right)+\frac{Y_{n+1}}{\eta_{n+1}}\left(\beta_{n+1}-\eta_{n+1}\right)^2 \right]}\\
&=z_n,
\end{split}
\end{equation}
which does not depend on the pair $(x,d)$
except the requirement $d\ne\gamma_n x$
(the latter is to exclude the deviation-position ratio $\frac00$, see~\eqref{eq:optstrat}).
Likewise, for all $x,d\in\bbR$, $d\ne\gamma_N x$,
the deviation-position ratio after the terminal trade equals
$$
\frac{d+\gamma_N\xi^*_N(x,d)}{x+\xi^*_N(x,d)}=\infty=z_N.
$$
It is worth noting that the process $z$ can take value $\infty$ also before the terminal time $N$
and it is even possible that $z$ takes finite values after being infinite
(see Section~\ref{sec:1go} for more detail).
\end{remark}

\section{Processes with independent multiplicative increments}\label{sec:PIMI}

In this section we restrict attention to resilience and price impact processes that satisfy\footnote{Recall that $\eta_n=\frac{\gamma_n}{\gamma_{n-1}}$, $n\in \Z$.}

\begin{description}
\item[(PIMI)]
for all $k\in \Z$ the random variables $\eta_{k+1}$ and $\beta_{k+1}$ are independent of $\mathcal F_k$.
\end{description}

In this case it turns out that the process $Y$ from Theorem~\ref{thm:mainres} is deterministic.

\begin{lemma}\label{prop:mii}
Assume (PIMI) and that for all $n\in \Z\cap(-\infty,N]$ we have $\beta_n, \gamma_n, \frac{1}{\gamma_n} \in \Linfm$ and
$E\left[\frac{\beta^2_{n}}{\eta_{n}}\right]<1$. 
Let $Y=(Y_n)_{n \in \Z \cap (-\infty,N]}$ be the process from Theorem \ref{thm:mainres} that is recursively defined by $Y_N=\frac{1}{2}$ and~\eqref{eq:defY}.
Then $Y$ is deterministic, 
$(0,\frac12]$-valued 
and satisfies the recursion
\begin{equation}\label{eq:Ydeterministic}
Y_n=E[\eta_{n+1}]Y_{n+1}-\frac{Y_{n+1}^2\left(E\left[\beta_{n+1}\right]-E\left[\eta_{n+1} \right]\right)^2}{Y_{n+1}E\left[\frac{(\beta_{n+1}-\eta_{n+1})^2}{\eta_{n+1}} \right]+\frac{1}{2}\left(1-E\left[\frac{\beta_{n+1}^2}{\eta_{n+1}}\right]\right)}, n\in \Z \cap (-\infty,N-1].
\end{equation}
Furthermore, formula~\eqref{eq:optstrat}
for optimal trade sizes
in the state $(x,d)\in\bbR^2$
takes the form
\begin{equation}\label{eq:13072019a1}
\xi^*_n(x,d)=
\frac{Y_{n+1}\left(E\left[\beta_{n+1}\right]-E\left[\eta_{n+1} \right]\right)}{Y_{n+1}E\left[\frac{(\beta_{n+1}-\eta_{n+1})^2}{\eta_{n+1}} \right]+\frac{1}{2}\left(1-E\left[\frac{\beta_{n+1}^2}{\eta_{n+1}}\right]\right)}
\left(x-\frac d{\gamma_n}\right)-\frac{d}{\gamma_n}, n \in \Z \cap (-\infty,N-1],
\end{equation}
and $\xi_N^*(x,d)=-x$.
\end{lemma}

\begin{proof}
Recursion~\eqref{eq:Ydeterministic} follows by a straightforward induction argument.
Formula~\eqref{eq:13072019a1} is an immediate consequence of the fact that $Y$ is deterministic.
\end{proof}

The particular case of (PIMI), where the sequences $(\eta_k)$ and $(\beta_k)$ are deterministic\footnote{It is worth noting that $(\gamma_k)$ can be random.}, deserves a separate treatment because, in this case, recursion~\eqref{eq:Ydeterministic} admits a closed-form expression. In fact, in this case, it is more convenient to work with the quantities
\begin{equation}\label{eq:05072020a1}
Z_k=\frac1{2Y_k},\quad k\in\bbZ\cap(-\infty,N],
\end{equation}
in place of $Y_k$, $k\in\bbZ\cap(-\infty,N]$.

\begin{corollary}\label{cor:05072020a1}
Assume that, for all $k\in\bbZ\cap(-\infty,N]$, $\gamma_k,\frac1{\gamma_k}\in L^{\infty-}$,
$\eta_k$ and $\beta_k$ are deterministic and $\beta_k^2<\eta_k$.
Let the (deterministic) sequence $Z=(Z_k)_{k\in\bbZ\cap(-\infty,N]}$ be defined by~\eqref{eq:05072020a1},
where the sequence $Y=(Y_k)_{k\in\bbZ\cap(-\infty,N]}$ is recursively defined by $Y_N=\frac12$ and~\eqref{eq:Ydeterministic}.
Then $Z$ is $[1,+\infty)$-valued and it holds
\begin{equation}\label{eq:05072020a2}
Z_k=\left(\prod_{i=k+1}^N \frac1{\eta_i}\right)+\sum_{j=k+1}^N \left(\,\prod_{i=k+1}^j \frac1{\eta_i}\right) \frac{(\eta_j-\beta_j)^2}{\eta_j-\beta_j^2}
,\quad k\in\bbZ\cap(-\infty,N],
\end{equation}
where $\sum_{N+1}^N:=0$, $\prod_{N+1}^N:=1$ (i.e., $Z_N=1$).
Furthermore, formula~\eqref{eq:13072019a1} for optimal trade sizes in the state $(x,d)\in\bbR^2$ takes the form
\begin{equation}\label{eq:05072020a3}
\xi^*_k(x,d)=
\frac{\eta_{k+1}-\beta_{k+1}}{\frac{(\eta_{k+1}-\beta_{k+1})^2}{\eta_{k+1}}+Z_{k+1}\left(1-\frac{\beta_{k+1}^2}{\eta_{k+1}}\right)}
\left(\frac d{\gamma_k}-x\right)-\frac{d}{\gamma_k},\quad k\in\bbZ\cap(-\infty,N-1],
\end{equation}
and $\xi_N^*(x,d)=-x$.
\end{corollary}

\begin{proof}
In the current setting, recursion~\eqref{eq:Ydeterministic} simplifies to $Y_N=\frac12$ and
$$
Y_k=\frac{\frac12\left(1-\frac{\beta^2_{k+1}}{\eta_{k+1}}\right)\eta_{k+1}Y_{k+1}}{Y_{k+1}\frac{(\eta_{k+1}-\beta_{k+1})^2}{\eta_{k+1}}+\frac12\left(1-\frac{\beta^2_{k+1}}{\eta_{k+1}}\right)},
\quad k\in\bbZ\cap(-\infty,N-1],
$$
which, for the sequence $Z$, yields $Z_N=1$ and
\begin{equation}\label{eq:05072020a4}
Z_k=\frac{(\eta_{k+1}-\beta_{k+1})^2}{\eta_{k+1}^2-\eta_{k+1}\beta_{k+1}^2}+\frac1{\eta_{k+1}}Z_{k+1},
\quad k\in\bbZ\cap(-\infty,N-1],
\end{equation}
and admits closed-form expression~\eqref{eq:05072020a2}.
The fact that $Z$ is $[1,+\infty)$-valued follows from the fact that $Y$ is $(0,\frac12]$-valued and~\eqref{eq:05072020a1}.
The last statement follows by a straightforward transformation in~\eqref{eq:13072019a1}.
\end{proof}

The formulas simplify even further when we additionally assume a constant order book depth $\gamma_k\equiv\gamma$, which means that $\eta_k\equiv1$.
The formula $\gamma_k\equiv\gamma$ is a slight abuse of our notation because in other places $\gamma$ denotes the whole sequence $(\gamma_k)$,
but we use $\gamma$ in place of $\gamma_k$ only in Corollary~\ref{cor:05072020a2}, and this does not cause any ambiguity in the sequel.

\begin{corollary}\label{cor:05072020a2}
Assume that, for all $k\in\bbZ\cap(-\infty,N]$, $\gamma_k=\gamma$ a.s.\ with some strictly positive $\bigcap_{k\in\bbZ}\cF_k$-measurable random variable $\gamma$ satisfying $\gamma,\frac1\gamma\in L^{\infty-}$.
In particular, $\eta_k=1$ a.s.\ for all $k\in\bbZ\cap(-\infty,N]$.
Further assume that the sequence $(\beta_k)_{k\in\bbZ\cap(-\infty,N]}$ is deterministic and $(0,1)$-valued.
Then we are in the situation of Corollary~\ref{cor:05072020a1},
formula~\eqref{eq:05072020a2} simplifies to
$$
Z_k=1+\sum_{j=k+1}^N \frac{1-\beta_j}{1+\beta_j},
\quad k\in\bbZ\cap(-\infty,N],
$$
where $\sum_{N+1}^N:=0$,
and formula~\eqref{eq:05072020a3} for optimal trade sizes in the state $(x,d)\in\bbR^2$ takes the form
\begin{align*}
\xi^*_k(x,d)
&=\frac1{1-\beta_{k+1}+(1+\beta_{k+1})Z_{k+1}}\left(\frac d{\gamma}-x\right)-\frac{d}{\gamma}\\
&=\frac1{2+(1+\beta_{k+1})\sum_{j=k+2}^N \frac{1-\beta_j}{1+\beta_j}}\left(\frac d{\gamma}-x\right)-\frac{d}{\gamma}
,\quad k\in\bbZ\cap(-\infty,N-1],
\end{align*}
and $\xi_N^*(x,d)=-x$.
\end{corollary}

\begin{proof}
The result follows from Corollary~\ref{cor:05072020a1} via straightforward calculations.
\end{proof}

In the next proposition we discuss the long-time limit
$\lim_{n\to-\infty}Y_n$ assuming (PIMI)
and a sort of time-homogeneity
(only for expectations).

\begin{propo}\label{propo:longtimelimMI}
Suppose that the assumptions of Lemma~\ref{prop:mii} hold true and that $\lbeta = E\left[ \beta_{n+1} \right]$, $\leta = E\left[ \eta_{n+1}\right]$ and $\lalpha =E \left[ \frac{\beta_{n+1}^2}{\eta_{n+1}} \right]$ do not depend on $n\in \Z\cap (-\infty,N-1]$. 
\begin{enumerate}
\item If $\lbeta = 1$, we have $\leta>1$, and
it holds for all $n\in \Z\cap (-\infty,N]$ that $Y_n = \frac{1}{2}$.
\item If $\leta\leq 1$, we have $\lbeta < 1$, and the sequence $Y=\left( Y_n \right)_{n\in \Z \cap (-\infty,N]}$ converges monotonically to $0$ as $n \to -\infty$.
\item If $\lbeta \neq 1$ and $\leta >1$, the sequence $Y=\left( Y_n \right)_{n\in \Z \cap (-\infty,N]}$ converges monotonically to 
\begin{equation}
\frac{ \frac{1}{2} \left( 1 - \lalpha \right) \left( \leta - 1 \right)}{\left( 1- \lalpha \right)  \left(\leta - 1\right) + \left( \lbeta - 1 \right)^2 } \in \left(0,\frac{1}{2}\right)
\end{equation}
as $n \to - \infty$. 
\end{enumerate}
\end{propo}

\paragraph{Discussion of Proposition~\ref{propo:longtimelimMI}}
Suppose that at time $n$ we have $x=1$ share to sell and the initial deviation is $d=0$.
The immediate selling of the share incurs the costs $\frac{\gamma_n}2$.
The optimal execution strategy produces the expected costs
$V_n(1,0)=\gamma_nY_n$ (recall~\eqref{eq:sol_valfct}).
So, in other words, the question about the long-time limit
$\lim_{n\to-\infty}Y_n$
is the question of how much better
in comparison to the immediate selling
we can perform if our time horizon is very big.

In general, dividing a large order into many small orders
and executing them in consecutive time points
can be profitable compared to the immediate execution
because of the following reasons:
\begin{enumerate}[(1)]
\item
the price impact process $\gamma$ penalizes trades at different times in a different way whenever $\gamma$ is nonconstant,
\item
the resilience process $\beta$ changes the deviation process $D$ between the trades whenever $\beta$ is not identically~$1$.
\end{enumerate}
From this viewpoint the claims of Proposition~\ref{propo:longtimelimMI},
which deals with the ``time-homogeneous in expectation (PIMI) case'',
are naturally interpreted as follows.
If the resilience is in expectation $1$ ($\lbeta=1$),
then neither of the above reasons suggests
dividing a large order into many small orders
(notice that, in this case, the price impact process $\gamma$
is increasing in average, as $\leta>1$).
We can asymptotically get rid of the execution costs
in the case of nonincreasing price impact
(in the sense $\leta\le1$).
Notice that, in this case, the price impact is allowed to be constant,
but we anyway profit from the resilience,
which, in expectation, drives the deviation back to zero between two trades
($\lbeta<1$).
Finally, in the remaining case of a nontrivial resilience and a geometrically increasing price impact
(in the sense $\lbeta\ne1$ and $\leta>1$)
we cannot fully get rid of the execution costs regardless of how big our time horizon is.

\begin{proof}[Proof of Proposition~\ref{propo:longtimelimMI}]
From \eqref{eq:Ydeterministic}, we have 
\begin{equation}\label{eq:Yrecmi}
Y_n = \leta Y_{n+1} - \frac{Y_{n+1}^2 \left(\lbeta - \leta \right)^2}{Y_{n+1}\left( \lalpha - 2 \lbeta + \leta \right) + \frac{1}{2}\left( 1 - \lalpha \right)},\quad n\in \Z \cap (-\infty,N-1] .
\end{equation}
Define $g\colon[0,\infty) \to \mathbb{R},$
\begin{equation}\label{eq:defgforfp}
g(y) = \leta y - \frac{y^2 \left(\lbeta - \leta \right)^2}{y\left( \lalpha - 2 \lbeta + \leta \right) + \frac{1}{2}\left( 1 - \lalpha \right)}, \qquad y\in [0,\infty).
\end{equation}
Note that $\lalpha < 1$ by assumption and that $\lalpha - 2 \lbeta + \leta \geq \frac{\left(\lbeta - \leta \right)^2}{\leta} \geq 0$ because $\frac{\lbeta^2}{\leta} \leq \lalpha$ by the Cauchy-Schwarz inequality.
Let $y\geq 0$. Then
\begin{equation*}
\begin{split}
g'(y) & = \leta - \left( \lbeta - \leta \right)^2 \frac{2y\left( y\left( \lalpha - 2\lbeta +\leta \right) + \frac{1}{2} \left( 1 -\lalpha \right) \right) - y^2 \left( \lalpha - 2\lbeta + \leta \right) }{\left( y\left( \lalpha - 2\lbeta +\leta \right) + \frac{1}{2} \left( 1 -\lalpha \right) \right)^2} \\
	& = \leta - \left( \lbeta - \leta \right)^2 \frac{y^2 \left( \lalpha - 2\lbeta +\leta \right) + y \left( 1 -\lalpha \right) }{\left( y\left( \lalpha - 2\lbeta +\leta \right) + \frac{1}{2} \left( 1 -\lalpha \right) \right)^2} .
\end{split}
\end{equation*}
Hence, $g'(y)>0$ is equivalent to 
\begin{equation*}
\leta \left( y\left( \lalpha - 2\lbeta +\leta \right) + \frac{1}{2} \left( 1 -\lalpha \right) \right)^2 > \left( \lbeta - \leta \right)^2 \left( y^2 \left( \lalpha - 2\lbeta +\leta \right) + y \left( 1 -\lalpha \right) \right).
\end{equation*} 
Divide by $\leta>0$ and note that $\frac{\left( \lbeta - \leta \right)^2}{\leta}= \frac{\lbeta^2}{\leta} - 2\lbeta + \leta$. This yields the equivalent statement 
\begin{equation*}
\begin{split}
0 & < y^2 \left( \lalpha - 2\lbeta +\leta \right)^2 + y  \left( \lalpha - 2\lbeta +\leta \right) \left( 1 -\lalpha \right) + \frac{\left( 1 -\lalpha \right)^2 }{4} - \frac{\left( \lbeta - \leta \right)^2}{\leta} y^2 \left( \lalpha - 2\lbeta +\leta \right) \\
	&  \quad - \frac{\left( \lbeta - \leta \right)^2}{\leta} y \left( 1 -\lalpha \right) \\
	& = y^2 \left( \lalpha - 2\lbeta +\leta \right) \left(\lalpha - \frac{\lbeta^2}{\leta} \right) + y \left(\lalpha - \frac{\lbeta^2}{\leta} \right) \left( 1 -\lalpha \right) + \frac{\left( 1 -\lalpha \right)^2 }{4} \\
	& = \left( y \left(\lalpha - \frac{\lbeta^2}{\leta} \right) + \frac{1-\lalpha}{2} \right)^2 + y^2 \left(\lalpha - \frac{\lbeta^2}{\leta} \right) \frac{\left( \lbeta - \leta \right)^2}{\leta}.
\end{split}
\end{equation*} 
Since $\lalpha <1$ and $\frac{\lbeta^2}{\leta} \leq \lalpha$, this always holds true for $y\geq 0$.
It follows that $g$ is strictly increasing on $[0,\infty)$.

Recall that $0<Y_n\le\frac12$
for all $n\in\bbZ\cap(-\infty,N-1]$
and $Y_N=\frac12$.
In particular, $Y_{N-1}\le Y_N$.
The recursion $Y_n=g(Y_{n+1})$,
$n\in\bbZ\cap(-\infty,N-1]$
(cf.\ \eqref{eq:Yrecmi} and~\eqref{eq:defgforfp}),
implies that the sequence $Y$ is nondecreasing.
Hence, the limit $\lim_{n\to-\infty}Y_n$
exists and belongs to $[0,\frac12]$.
Moreover, it is the largest fixed point of $g$ in $[0,\frac{1}{2}]$.
Indeed, since $g$ is increasing,
for the largest fixed point $\bar y$ of $g$ in $[0,\frac{1}{2}]$,
we have that $y\ge\bar y$ implies $g(y)\ge g(\bar y)=\bar y$.
Hence, $\bar y$ is a lower bound of $Y$.
We obtain that $\lim_{n\to-\infty}Y_n\ge\bar y$
and is a fixed point of $g$, which means that
$\lim_{n\to-\infty}Y_n=\bar y$.

\begin{enumerate}
\item
Suppose that $\lbeta = 1$.
The claim that $\leta>1$ follows from $\frac{\lbeta^2}{\leta}\le\alpha<1$.
A direct calculation shows that $g\left(\frac{1}{2}\right)=\frac{1}{2}$.
Since $Y_N=\frac{1}{2}$, it follows that $Y_n  = \frac{1}{2}$ for all $n\in \Z\cap (-\infty,N]$.

\item
Suppose that $\leta \leq 1$. First notice that $\lbeta^2 \leq \leta \lalpha < \leta \leq 1$ and hence $\lbeta < 1$.
Now it follows from~\eqref{eq:defgforfp} that
for all $y>0$ we have $g(y)<y$.
This yields that $0$ is the only fixed point of $g$ on $[0,\infty)$
and hence $\lim_{n\to-\infty}Y_n=0$.

\item
Suppose that $\lbeta \neq 1$ and $\leta>1$.
In this case
\begin{equation}\label{eq:def:bary}
\bar y = \frac{ \frac{1}{2} \left( 1 - \lalpha \right) \left( \leta - 1 \right)}{\left( 1- \lalpha \right)  \left(\leta - 1\right) + \left( \lbeta - 1 \right)^2 } \in \left(0,\frac{1}{2}\right)
\end{equation}
is a further fixed point of $g$ and the only one in $(0,\infty)$. Indeed, 
for $y\in (0,\infty)$ the condition $g( y)=y$ is equivalent to
\begin{equation}\label{eq:fpcondnotzero2}
y\left( \left(\lbeta - \leta \right)^2 - \left(\leta - 1\right) \left( \lalpha - 2 \lbeta + \leta\right) \right) = \frac{1}{2} \left( 1 - \lalpha \right) \left( \leta - 1 \right) . 
\end{equation}
From the fact that
\begin{equation*}
\left(\lbeta - \leta \right)^2 - \left(\leta - 1\right) \left( \lalpha - 2 \lbeta + \leta\right) = \left( 1- \lalpha \right)  \left(\leta - 1\right) + \left( \lbeta - 1 \right)^2
>\left( 1- \lalpha \right)  \left(\leta - 1\right)>0
\end{equation*}
we deduce~\eqref{eq:def:bary}, which completes the proof.
\end{enumerate}
\end{proof}

The following lemma provides an example where the process $Y=\left( Y_n\right)_{n \in \Z\cap(-\infty,N]}$ defined by $Y_N=\frac{1}{2}$ and \eqref{eq:defY} does not converge. In this example the price impact process $\gamma$ is a submartingale (cf.\ the discussion following Proposition~\ref{prop:08062019a1}).

\begin{lemma}\label{lemma:Ydoesnotconverge}
Suppose that the assumptions of Lemma~\ref{prop:mii} hold true. 
Let $\lbeta_1,\lbeta_2,\leta_1,\leta_2 \in (0,\infty)$ and $\lalpha_1, \lalpha_2 \in (0,1)$ such that for all $k\in \N_0$ it holds $\lbeta_1=E\left[ \beta_{N-2k-1} \right]=1$, $\lbeta_2=E\left[ \beta_{N-2k} \right]\neq 1$, $\leta_1=E\left[ \eta_{N-2k-1} \right] $, $\leta_2=E\left[ \eta_{N-2k} \right] >1$, $\lalpha_1=E\left[ \frac{\beta_{N-2k-1}^2}{\eta_{N-2k-1}} \right]$ and $\lalpha_2=E\left[ \frac{\beta_{N-2k}^2}{\eta_{N-2k}} \right]$. 

Then, $\gamma$ is a submartingale and $Y=\left( Y_n\right)_{n \in \Z\cap(-\infty,N]}$ does not converge as $n \to -\infty$. In particular, the sequence $Y$ is not monotone.
\end{lemma}

\begin{proof}
Note first that $\lbeta_1=1$ and $\lalpha_1<1$ imply that $\leta_1>1$ by the Cauchy-Schwarz inequality. It follows from $1 < \leta_1 = E\left[ \eta_{N-2k-1} \right] = E_{N-2k-2}\left[ \eta_{N-2k-1} \right] = E_{N-2k-2}\left[ \frac{\gamma_{N-2k-1}}{\gamma_{N-2k-2}} \right] = \frac{1}{\gamma_{N-2k-2}} E_{N-2k-2}\left[ \gamma_{N-2k-1} \right] $ and $1 < \leta_2 = \frac{1}{\gamma_{N-2k-1}} E_{N-2k-1}\left[ \gamma_{N-2k} \right] $ for all $k \in \N_0$ that $\gamma$ is a submartingale.

For $i \in \{1,2\}$, denote by $g_i$ the function defined by \eqref{eq:defgforfp} with $\lbeta=\lbeta_i$, $\leta=\leta_i$ and $\lalpha=\lalpha_i$. Recall that $g_1,g_2$ are strictly increasing and note that for $k \in \N_0$, we have $Y_{N-2k-2}=g_1\left(Y_{N-2k-1}\right)$ and $Y_{N-2k-1}=g_2\left(Y_{N-2k}\right)$. 
Furthermore, the equations $g_i(y)=y$, $i\in\{1,2\}$,
are (non-degenerate) quadratic ones,
hence the functions $g_i$ have at most two fixed points.
We conclude that the only fixed points of $g_1$ are $0$ and $\frac{1}{2}$,
and the only fixed points of $g_2$ are given by $0$ and $\lYtwo \in \left(0,\frac{1}{2}\right)$ from~\eqref{eq:def:bary}. 
We also notice that
$g_1(y)>y$ for $y\in\left(0,\frac{1}{2}\right)$.

We prove by induction that $Y_{N-m}> \lYtwo$ for all $m \in \N_0$. The case $m=0$ is clear. For the induction step $\N_0 \ni m \to m+1 \in \N$, if $m$ is even, we have $Y_{N-m-1}=g_2\left( Y_{N-m} \right) > g_2\left( \lYtwo\right) = \lYtwo$. If $m$ is odd, it holds $Y_{N-m-1}=g_1\left( Y_{N-m}\right) > g_1\left(\lYtwo\right) > \lYtwo$.

It can further be proven inductively that $Y_{N-m}\geq Y_{N-m-2}$ for all $m \in\N_0$ since $g_1,g_2$ are increasing and $Y_{N-2}\leq\frac{1}{2}=Y_N$.

Therefore, the subsequences $\left( Y_{N-2k}\right)_{k\in\N_0}$ and $\left( Y_{N-2k-1}\right)_{k\in\N_0}$ of $Y$ are decreasing in $k\in\N_0$ and bounded from below by $\lYtwo$, which implies that the limits $\lYe = \lim_{k\to\infty}Y_{N-2k} \geq \lYtwo$ and $\lYo = \lim_{k\to\infty}Y_{N-2k-1} \geq \lYtwo$ exist. 
Taking limits on both sides of $Y_{N-2k-1}=g_2\left(Y_{N-2k}\right)$, we obtain $\lYo=g_2\left(\lYe\right)$ by continuity of $g_2$. Similarly, it holds that $\lYe=g_1\left(\lYo\right)$. 
Now, if $\lYe$ and $\lYo$ were equal, then $\lYe=\lYo$ would be a common fixed point of $g_1$ and $g_2$ and hence $0$, which is a contradiction to $\lYe\geq \lYtwo >0$. 
We thus conclude that $Y$ does not converge. 
\end{proof}

\section{Round trips}\label{sec:roundtrips}

Let $n\in\bbZ\cap(-\infty,N-1]$.
Execution strategies in $\mathcal A_n(0)$ are called round trips.
It follows from Theorem~\ref{thm:mainres} that
if initially the agent has no position in the asset, i.e.,
$x=0$ at time $ n \in \Z \cap (-\infty,N]$,
then the minimal costs amount to
\begin{equation}\label{eq:30052019a1}
V_n(0,d)=\frac{d^2}{\gamma_n}\left(Y_n-\frac{1}{2}\right)
\end{equation}
for all $d\in \R$. In particular, it holds that $V_n(0,0)=0$, i.e.,
without initial deviation of the price process
the agent cannot make profits in expectation.
In other words, there are no profitable round trips
whenever $d=0$. The existence of profitable round trips is sometimes also referred to as price manipulation
(see, e.g.,
\cite{alfonsi2010boptimal},
\cite{gatheral2010no}
or \cite{huberman2004price}).
In this regard, if there is no initial deviation of the price process (i.e., $d=0$), then our model does not admit price manipulation.

Below we study existence of profitable round trips
when the price of a share deviates from the unaffected price,
i.e., it holds $d\ne0$.
We thus assume $d\ne0$ in this section.
Recall from \eqref{eq:sol_valfct} that the random variable $Y_n$ is $(0,\frac12]$-valued.
Together with~\eqref{eq:30052019a1},
this implies the following classification:
\begin{itemize}
\item
on $\{Y_n<\frac12\}$ there exist profitable round trips,

\item
on $\{Y_n=\frac12\}$ there are no profitable round trips.
\end{itemize}
Thus, the question reduces to finding
a tractable description of the event
$\{Y_n=\frac12\}$.
We first characterize this event
in Proposition~\ref{prop:31052019a1}
and discuss several consequences of this characterization.
The proof of Proposition~\ref{prop:31052019a1}
is postponed to Subsection~\ref{subsec:rt_proofs}.

\begin{propo}\label{prop:31052019a1}
Let the assumptions of Theorem~\ref{thm:mainres} be satisfied.
Then we have
$$
\left\{Y_n=\frac{1}{2}\right\} =
\left\{ E_n\left[ Y_{n+1}\right] = \frac12, E_n\left[\beta_{n+1}\right] =1\right\},
\quad n \in \Z \cap (-\infty,N-1],
$$
where here and below we understand
the equalities for events up to $P$-null sets.
\end{propo}

\begin{corollary}\label{cor:30052019a1}
Under the assumptions of Theorem~\ref{thm:mainres}
it holds
$$
\left\{ Y_{N-1}=\frac{1}{2} \right\}=
\{ E_{N-1}\left[ \beta_{N} \right]=1 \}.
$$
\end{corollary}

\begin{proof}
The result is immediate because
$Y_N=\frac12$.
\end{proof}

\begin{corollary}\label{propo:propY}
Under the assumptions of Theorem~\ref{thm:mainres} we have the following inclusions for $n \in \Z \cap (-\infty,N-1]$:
\begin{enumerate}
\item\label{propYitem2}
$\{Y_n=\frac{1}{2}\} \subseteq \{ Y_{n+1}=\frac{1}{2} \}$
(equivalently,
$\{Y_{n+1}<\frac{1}{2}\} \subseteq \{ Y_{n}<\frac{1}{2} \}$) and

\item\label{propYitem3}
$\{Y_n=\frac{1}{2}\}
\subseteq \{E_n\left[ \beta_{n+1} \right]=1\}
\subseteq \{E_n\left[ \beta_{n+1} \right]\ge1\}
\subseteq \{E_n\left[ \eta_{n+1}\right]>1\}$
(equivalently,
$\{E_n\left[ \eta_{n+1}\right]\le1\}
\subseteq \{E_n\left[ \beta_{n+1} \right]<1\}
\subseteq\{E_n\left[ \beta_{n+1} \right]\ne1\}
\subseteq\{Y_n<\frac{1}{2}\}$).
\end{enumerate}
\end{corollary}
The proof of Corollary~\ref{propo:propY} is given in Subsection~\ref{subsec:rt_proofs}.

\paragraph{Discussion}
In the literature on optimal execution
it is often assumed that the resilience process
$\beta$ takes values in $(0,1)$.
In this case we always have profitable round trips
whenever $d\ne0$,
as we know that the deviation
will go towards zero due to the resilience
and we can make use of it in constructing
a profitable round trip 
(cf.\ Remark~8.2 in \cite{fruth2014optimal}
and the discussion after Model~8.3
in \cite{fruth2019optimal}).
Formally, this fact follows from Corollary~\ref{propo:propY}. 
A natural generalization of this fact
to the case of (only) positive $\beta$
is the inclusion
$\{E_n[\beta_{n+1}]\ne1\}
\subseteq\{Y_n<\frac12\}$ 
(again Corollary~\ref{propo:propY}).
The intuition is that 
on the event $\{E_n[\beta_{n+1}]\ne1\}$
we ``expect'' in which direction
the deviation will go in the absence of trading.
A new qualitative effect in our setting is
that the situation of nonexistence
of profitable round trips is possible.
The previous discussion explains
that we necessarily need to be on the event
$\{E_n[\beta_{n+1}]=1\}$
for the non-existence of profitable round trips.
A somewhat unexpected effect is, however,
that the inclusion
$\{Y_n=\frac{1}{2}\}
\subseteq \{E_n\left[ \beta_{n+1} \right]=1\}$
can be strict and hence there might exist
profitable round trips on the event
$\{E_n\left[ \beta_{n+1} \right]=1\}$ 
(see Examples \ref{ex:11072019a1} and~\ref{ex:11072019a2} below
for a more precise discussion). 
In particular,
we cannot distinguish
$Y_n=\frac12$ from $Y_n<\frac12$
on the basis of $E_n[\beta_{n+1}]$ alone,
and, indeed, the exact characterization of the event
$\{Y_n=\frac12\}$
also includes $E_n[Y_{n+1}]$
(see Proposition~\ref{prop:31052019a1}).

In more detail, we have the following picture.
At time $N-1$ we distinguish between
$Y_{N-1}=\frac12$ from $Y_{N-1}<\frac12$
on the basis of $E_{N-1}[\beta_{N}]$ alone
(Corollary~\ref{cor:30052019a1}).
To discuss the step $n+1\to n$ we consider
the partition of $\Omega$ into two disjoint events
(in~$\cF_n$)
\begin{equation}\label{eq:30052019a2}
\Omega
=
\left\{E_n[Y_{n+1}]<\frac12\right\}
\sqcup
\left\{E_n[Y_{n+1}]=\frac12\right\}
=:A_n\sqcup B_n.
\end{equation}
On $A_n$ there always exist profitable round trips
when we start at time $n$,
while on $B_n$ we distinguish
between the nonexistence and the existence
of profitable round trips on the basis
of whether $E_n[\beta_{n+1}]=1$
or $E_n[\beta_{n+1}]\ne1$ holds
(Proposition~\ref{prop:31052019a1}).

A special case, where
we obtain an explicit criterion
to distinguish between
$Y_n=\frac12$ and $Y_n<\frac12$
for all $n\in\bbZ\cap(-\infty,N-1]$
only in terms of the process $\beta$
is the case of processes with independent multiplicative
increments of Section~\ref{sec:PIMI}:

\begin{corollary}\label{cor:11072019a1}
Let the assumptions of Lemma~\ref{prop:mii} be in force.
We define
$$
n_0=N\wedge\inf\{n\in\bbZ\cap(-\infty,N-1]:
E[\beta_k]=1\text{ for all }k\in\bbZ\cap[n+1,N]\}
$$
($\inf\emptyset=\infty$)
and notice that $n_0\in(\bbZ\cup\{-\infty\})\cap[-\infty,N]$.
Then, for the (deterministic) process $Y$, we have
\begin{itemize}
\item
$Y_n<\frac12$ for $n\in\bbZ\cap(-\infty,n_0)$,

\item
$Y_n=\frac12$ for $n\in\bbZ\cap[n_0,N]$.
\end{itemize}
\end{corollary}

\begin{proof}
The result follows from the previous discussion
and the fact that, by Lemma~\ref{prop:mii},
the process $Y$ is deterministic.
\end{proof}

The next proposition contains
a sufficient condition
for existence of profitable round trips,
which is expressed in different terms.

\begin{propo}\label{prop:30052019a1}
Under the assumptions of Theorem~\ref{thm:mainres}
for all $n \in \Z \cap (-\infty,N-1]$ it holds
$$
\left\{Y_n=\frac12\right\}
\subseteq
\left\{\min_{k\in\{n+1,\ldots,N\}}E_n(\gamma_k)\ge\gamma_n\right\}
$$
(equivalently,
$\{\min_{k\in\{n+1,\ldots,N\}}E_n(\gamma_k)<\gamma_n\}
\subseteq
\{Y_n<\frac12\}$).
\end{propo}

\begin{proof}
While the result can be again inferred from
the characterization of the event $\{Y_n=\frac12\}$
in Proposition~\ref{prop:31052019a1},
the shortest proof is to recall that
$Y_n<\frac12$
on the event
$\{\min_{k\in\{n+1,\ldots,N\}}E_n(\gamma_k)<\gamma_n\}$
due to Remark~\ref{rem:ubY}.
\end{proof}

We now discuss the inclusion $\{Y_n=\frac12\}\subseteq\{E_n[\beta_{n+1}]=1\}$ in more detail.
First we present a simple example, where for $n=N-2$ this inclusion is strict (cf.\ with Corollary~\ref{cor:30052019a1}).

\begin{ex}\label{ex:11072019a1}
We take any deterministic sequences $\beta$ and $\gamma$
with $\beta_N\ne1$ and $\beta_{N-1}=1$
that satisfy the assumptions of Theorem~\ref{thm:mainres}.
Then the process $Y$ is deterministic.
Corollary~\ref{cor:30052019a1} implies that $Y_{N-1}<\frac12$.
Hence, by Corollary~\ref{propo:propY}, $Y_{N-2}<\frac12$.
We thus have
$$
\left\{Y_{N-2}=\frac12\right\}=\emptyset\subsetneq\Omega=\{E_{N-2}[\beta_{N-1}]=1\}.
$$
In other words, for $d\ne0$, we have profitable round trips when we start at time $N-2$, although $E_{N-2}[\beta_{N-1}]=1$.
This is not surprising in this example, as we see that profitable round trips are already present when we start at time $N-1$
($Y_{N-1}<\frac12$, which is caused by $\beta_N\ne1$).
One might, therefore, intuitively expect that here all round trips do not contain a trade at time $N-2$,
but this is not the case!
If $d\ne0$, then we have for the (here, deterministic) optimal strategy $\xi^*(0,d)$ of~\eqref{eq:optstrat} that
$\xi^*_{N-2}(0,d)\ne0$. Indeed, a straightforward calculation using~\eqref{eq:optstrat}
and the fact that $\beta$, $\eta$, $Y$ are deterministic and $\beta_{N-1}=1$ reveals that
$\xi^*_{N-2}(0,d)=0$ if and only if it holds
$(\frac12-Y_{N-1})(1-\frac1{\eta_{N-1}})=0$,
but the latter is not true in this example because
$Y_{N-1}<\frac12$ and $\frac1{\eta_{N-1}}=\frac{\beta^2_{N-1}}{\eta_{N-1}}<1$
(recall the assumptions of Theorem~\ref{thm:mainres}).
\end{ex}

Example~\ref{ex:11072019a1} raises the question
of whether profitable round trips for $d\ne0$ with starting time $n\in\bbZ\cap(-\infty,N-2]$
can occur on the event
$\bigcap_{k=n}^{N-1}\{E_k[\beta_{k+1}]=1\}$.
Corollary~\ref{cor:11072019a1} implies that this is impossible in the framework of (PIMI)
(let alone with deterministic $\beta$ and~$\gamma$).
But, in general, such a phenomenon is possible,
and we present a specific example after the following lemma.

\begin{lemma}\label{lem:12072019a1}
Let the assumptions of Theorem~\ref{thm:mainres} be in force and let $n \in \Z \cap (-\infty,N-1]$.

\smallskip\noindent
(i) We have
\begin{equation}\label{eq:11072019a1}
\left\{Y_n=\frac12\right\}\subseteq\bigcap_{k=n}^{N-1}\{E_k[\beta_{k+1}]=1\}.
\end{equation}

\smallskip\noindent
(ii) The inclusion in~\eqref{eq:11072019a1} is strict
(in the sense that the set difference has positive $P$-probability)
if and only if
\begin{equation}\label{eq:11072019a2}
\bigcap_{k=n}^{N-1}\{E_k[\beta_{k+1}]=1\}\notin\ol\cF_n,
\end{equation}
where $\ol\cF_n=\sigma(\cF_n\cup\cN)$ with $\cN=\{A\in\cF:P(A)=0\}$.
\end{lemma}

\begin{proof}
Inclusion~\eqref{eq:11072019a1} follows from Corollary~\ref{propo:propY}.
Clearly, under~\eqref{eq:11072019a2}, the inclusion is strict, as $\{Y_n=\frac12\}\in\cF_n$.
It remains to prove that, if there is $A_n\in\cF_n$,
which is (up to a $P$-null set) equal to
$\bigcap_{k=n}^{N-1}\{E_k[\beta_{k+1}]=1\}$,
then $Y_n=\frac12$ a.s.\ on $A_n$.

First, Corollary~\ref{cor:30052019a1} yields $Y_{N-1}=\frac12$ a.s.\ on $A_n$.
In the case $n=N-1$ this concludes the proof.
Let $n\le N-2$. As $A_n\in\cF_n\subseteq\cF_{N-2}$, we get
$E_{N-2}[Y_{N-1}]=\frac12$ a.s.\ on $A_n$.
Proposition~\ref{prop:31052019a1} now yields $Y_{N-2}=\frac12$ a.s.\ on $A_n$.
In the case $n=N-2$ this concludes the proof.
If $n\le N-3$, we obtain the result by iterating the same procedure.
\end{proof}

We, finally, present a specific example, where for $n=N-2$ the inclusion in~\eqref{eq:11072019a1} is strict,
or, in other words, $P(Y_{N-2}<\frac12,E_{N-2}[\beta_{N-1}]=E_{N-1}[\beta_N]=1)>0$
(recall the discussion following Example~\ref{ex:11072019a1}).

\begin{ex}\label{ex:11072019a2}
Take arbitrary $a,p\in(0,1)$.
Let $\cF_n=\{\emptyset,\Omega\}$ for $n \in \Z\cap(-\infty,N-2]$, $\cF_{N-1}=\cF_N=\sigma(\beta_{N-1})$
with $\beta_{N-1}$ being distributed according to
$P(\beta_{N-1}=1)=1-p$ and $P(\beta_{N-1}=1\pm a)=p/2$.
We set $\beta_N=\beta_{N-1}$ and choose any process $\gamma$
satisfying the assumptions of Theorem~\ref{thm:mainres}
(e.g., one can easily take deterministic~$\gamma$).
Then $E_{N-2}[\beta_{N-1}]=E[\beta_{N-1}]=1$, hence
$$
\{E_{N-2}[\beta_{N-1}]=1\}\cap\{E_{N-1}[\beta_{N}]=1\}
=\{E_{N-1}[\beta_{N}]=1\}=\{\beta_N=1\},
$$
which is an event of probability $1-p\in(0,1)$.
We thus obtain~\eqref{eq:11072019a2} for $n=N-2$.
By Lemma~\ref{lem:12072019a1}, the inclusion in~\eqref{eq:11072019a1} for $n=N-2$ is strict.
As a result, we get $P(Y_{N-2}<\frac12,E_{N-2}[\beta_{N-1}]=E_{N-1}[\beta_N]=1)>0$,
as required.
\end{ex}

\subsection{Proofs of Proposition~\ref{prop:31052019a1} and Corollary \ref{propo:propY}}\label{subsec:rt_proofs}

\begin{proof}[Proof of Proposition~\ref{prop:31052019a1}]
Throughout the proof fix $n \in \Z \cap (-\infty,N-1]$.
Let $\nu= \frac{1}{2} - \left( \frac{1}{2} - Y_{n+1}\right) \frac{\beta_{n+1}^2}{\eta_{n+1}}$. Rewriting the definition of $Y_n$, we obtain 
\begin{equation*}
\begin{split}
Y_n & = E_n[\eta_{n+1}Y_{n+1}]-\frac{\left(E_n\left[Y_{n+1}\beta_{n+1} \right]\right)^2 - 2 E_n\left[Y_{n+1}\beta_{n+1} \right]E_n\left[Y_{n+1}\eta_{n+1} \right] + \left(E_n\left[Y_{n+1}\beta_{n+1} \right]\right)^2}{E_n\left[ \nu - 2Y_{n+1}\beta_{n+1}+Y_{n+1}\eta_{n+1}\right]} \\
	& = \frac{E_n\left[\nu\right] E_n\left[ \nu - 2Y_{n+1}\beta_{n+1}+Y_{n+1}\eta_{n+1}\right] -  \left(E_n\left[ \nu - Y_{n+1}\beta_{n+1} \right]\right)^2 }{E_n\left[ \nu - 2Y_{n+1}\beta_{n+1}+Y_{n+1}\eta_{n+1}\right]} \\
	& = \frac{1}{2} - E_n\left[ \left( \frac{1}{2} - Y_{n+1}\right) \frac{\beta_{n+1}^2}{\eta_{n+1}} \right] - \frac{\gamma_n}{a_n} \left(  \frac{1}{2} - E_n\left[\left( \frac{1}{2} - Y_{n+1}\right) \frac{\beta_{n+1}^2}{\eta_{n+1}}\right] - E_n\left[ Y_{n+1}\beta_{n+1}\right] \right)^2
\end{split}
\end{equation*}
with $a_n$ from \eqref{eq:abc}. Since $\eta_{n+1},\gamma_n, a_n >0$ and $Y_{n+1}\leq \frac{1}{2}$ a.s., it now follows that 
\begin{equation}\label{eq:31052019a1}
\left\{Y_n=\frac{1}{2}\right\} = \left\{ E_n\left[ \left(\frac{1}{2}-Y_{n+1}\right) \frac{\beta_{n+1}^2}{\eta_{n+1}} \right] = 0, E_n\left[ Y_{n+1}\beta_{n+1}\right] = \frac{1}{2} \right\}.
\end{equation}
Let $C_n = \left\{E_n\left[\left(\frac{1}{2}-Y_{n+1}\right)
\frac{\beta_{n+1}^2}{\eta_{n+1}} \right]=0\right\}$ and 
denote $B_n = \left\{E_n[Y_{n+1}]=\frac12\right\}$ as before. We show that $C_n=B_n$. For the inclusion $C_n \supseteq B_n$ note first that 
\begin{equation}\label{eq:Yequnderint}
\int_{\left\{E_n\left[Y_{n+1}\right]=\frac{1}{2}\right\}} Y_{n+1}\,dP
=
\int_{\left\{E_n\left[Y_{n+1}\right]=\frac{1}{2}\right\}}
E_n\left[Y_{n+1}\right]\,dP
=
\int_{\left\{E_n\left[Y_{n+1}\right]=\frac{1}{2}\right\}} \frac12\,dP
\end{equation}
and hence that $Y_{n+1}=\frac{1}{2}$ on $B_n$. This together with the fact that 
$B_n \in \mathcal{F}_n$  
implies
\begin{equation*}
\begin{split}
1_{B_n} E_n\left[\left(\frac{1}{2}-Y_{n+1}\right) \frac{\beta_{n+1}^2}{\eta_{n+1}} \right] 
 & = E_n\left[ 1_{B_n}\left(\frac{1}{2}-Y_{n+1}\right) \frac{\beta_{n+1}^2}{\eta_{n+1}} \right] 
= 0.
\end{split}
\end{equation*}
To prove $C_n\subseteq B_n$, observe that $C_n \in \cF_n$ and that 
\begin{equation*}
C_n \subseteq \left\{\left(\frac{1}{2}-Y_{n+1}\right)
\frac{\beta_{n+1}^2}{\eta_{n+1}} =0\right\} 
=  \left\{Y_{n+1}=\frac{1}{2}\right\}
\end{equation*} 
(by an argument similar to \eqref{eq:Yequnderint})  
since $\beta_{n+1},\eta_{n+1}>0$ and $Y_{n+1}\leq \frac{1}{2}$ a.s. It thus holds that 
\begin{equation*}
1_{C_n} E_n\left[ Y_{n+1} \right] 
 = E_n\left[ 1_{C_n} Y_{n+1}\right] 
 = 1_{C_n} \frac{1}{2}.
\end{equation*}
From $C_n = B_n$ together with \eqref{eq:31052019a1} we obtain 
\begin{equation*}
\left\{Y_n=\frac{1}{2}\right\} = \left\{ E_n[Y_{n+1}]=\frac12, E_n\left[ Y_{n+1}\beta_{n+1}\right] = \frac{1}{2} \right\}.
\end{equation*}
Furthermore, we have 
\begin{equation*}
1_{B_n} E_n\left[ Y_{n+1}\beta_{n+1} \right] = E_n\left[ 1_{B_n} Y_{n+1}\beta_{n+1} \right] = 1_{B_n} \frac{1}{2}E_n\left[ \beta_{n+1} \right] ,
\end{equation*}
and hence  
\begin{equation*}
\left\{Y_n=\frac{1}{2}\right\} = \left\{ E_n[Y_{n+1}]=\frac12, E_n\left[ \beta_{n+1}\right] = 1 \right\}.
\end{equation*}
\end{proof}

\begin{proof}[Proof of Corollary \ref{propo:propY}]
We fix $n \in \Z \cap (-\infty,N-1]$.
\begin{enumerate}
\item
The claim follows from
$$
\left\{Y_n=\frac{1}{2}\right\}
\subseteq\left\{E_n\left[Y_{n+1}\right]=\frac{1}{2}\right\}
\subseteq\left\{Y_{n+1}=\frac{1}{2}\right\},
$$
where the first inclusion is immediate from
Proposition~\ref{prop:31052019a1}
and the second one follows from the facts
that $Y_{n+1}\le\frac12$ a.s.\ and \eqref{eq:Yequnderint}.

\item
Due to Proposition~\ref{prop:31052019a1}
only the inclusion
$\{ E_n\left[ \beta_{n+1} \right]\ge1 \} \subseteq \{E_n\left[ \eta_{n+1}\right]>1 \}$
needs to be proved.
By the Cauchy-Schwarz inequality and the assumption
$E_n\left[\frac{\beta^2_{n+1}}{\eta_{n+1}}\right]<1$ a.s.\ we get
\begin{equation*}
\left( E_n\left[ \beta_{n+1} \right] \right)^2
\leq E_n\left[ \frac{\beta_{n+1}^2}{\eta_{n+1}} \right] E_n\left[ \eta_{n+1}\right]
< E_n\left[ \eta_{n+1}\right]\;\;\text{a.s.},
\end{equation*}
which implies the claim.
\end{enumerate}
\end{proof}

\section{Closing the position in one go}\label{sec:1go}

Let the assumptions of
Theorem~\ref{thm:mainres}
be in force.
Let $n\in\bbZ\cap(-\infty,N-1]$.
We now study when
$\xi^*_n(x,d)=-x$
for all $x,d\in\bbR$,
i.e., when it is optimal to close
the whole position at time $n<N$.

Recall that, for each $x,d\in\bbR$,
a version of the optimal trade
$\xi^*_n(x,d)$
(which is defined up to a $P$-null set)
is given by the right-hand side of~\eqref{eq:optstrat}.
We choose the versions in such a way that the random field
$(x,d)\mapsto\xi^*_n(x,d)$ is continuous
(the most natural choice in view of~\eqref{eq:optstrat}).
Then we have
$$
\{\xi^*_n(x,d)=-x\;\forall x,d\in\bbR\}
=
\{\xi^*_n(x,d)=-x\;\forall x,d\in\bbQ\}
=
\bigcap_{x,d\in\bbQ} \{\xi^*_n(x,d)=-x\},
$$
hence
$\{\xi^*_n(x,d)=-x\;\forall x,d\in\bbR\}$
is an $\cF_n$-measurable event
(as a countable intersection of such events).

\begin{lemma}\label{lem:01062019a1}
Let $n\in\bbZ\cap(-\infty,N-1]$.
Under the assumptions of Theorem~\ref{thm:mainres}
we have
\begin{equation}\label{eq:01062019a1}
\{\xi^*_n(x,d)=-x\;\forall x,d\in\bbR\}=
\left\{E_n\left[
\left(Y_{n+1}-\frac12\right)\frac{\beta_{n+1}^2}{\eta_{n+1}}
-Y_{n+1}\beta_{n+1}+\frac12
\right]=0\right\},
\end{equation}
up to a $P$-null set.
\end{lemma}

\begin{proof}
The result follows from~\eqref{eq:optstrat}
via a straightforward calculation.
\end{proof}

The next result presents a relation
between the previously studied question
of nonexistence of profitable round trips for $d\ne0$
and the currently studied question
of closing the position in one go.

\begin{propo}\label{prop:01062019a1}
Let $n\in\bbZ\cap(-\infty,N-1]$.
Under the assumptions of Theorem~\ref{thm:mainres}
we have
\begin{enumerate}
\item\label{it:01062019a1}
$\{Y_n=\frac12\}
\subseteq
\{\xi^*_n(x,d)=-x\;\forall x,d\in\bbR\}$.

\item\label{it:01062019a2}
$\{Y_n=\frac12\}=
\{\xi^*_n(x,d)=-x\;\forall x,d\in\bbR\}\cap
\{E_n[Y_{n+1}]=\frac12\}$.
\end{enumerate}
\end{propo}

It is worth noting that the inclusion in part~\ref{it:01062019a1}
can be strict
in the sense that the set difference can be non-negligible, i.e., with positive probability there are profitable round trips at time $n$ for $d\neq 0$ and still it is optimal to close the whole position at time $n$
(see Example~\ref{ex:01062019a1} below).

\begin{proof}
1. Recall that by Proposition~\ref{prop:31052019a1}
and Corollary~\ref{propo:propY} we have
$$
\left\{Y_n=\frac{1}{2}\right\}
=\left\{ E_n\left[ Y_{n+1}\right] = \frac12, E_n\left[\beta_{n+1}\right] =1\right\}
\subseteq\left\{Y_{n+1}=\frac12\right\}.
$$
In particular, on the event
$\{Y_n=\frac12\}\in\cF_n$ it holds
$Y_{n+1}=\frac12$ and $E_n[\beta_{n+1}]=1$,
which implies that on the event
$\{Y_n=\frac12\}\in\cF_n$ we have
$$
E_n\left[
\left(Y_{n+1}-\frac12\right)\frac{\beta_{n+1}^2}{\eta_{n+1}}
-Y_{n+1}\beta_{n+1}+\frac12
\right]=0.
$$
Lemma~\ref{lem:01062019a1} now yields the claim.

\smallskip
2. The inclusion ``$\subseteq$'' follows from
the previous part together with
Proposition~\ref{prop:31052019a1}.
To prove the reverse inclusion
``$\supseteq$'' we first note that
\begin{equation}\label{eq:02062019a1}
\left\{E_n[Y_{n+1}]=\frac12\right\}
\subseteq
\left\{Y_{n+1}=\frac12\right\}
\end{equation}
because $Y_{n+1}\le\frac12$ a.s.
It follows from
\eqref{eq:01062019a1} and~\eqref{eq:02062019a1}
that on the $\cF_n$-measurable set
$$
A_n:=\{\xi^*_n(x,d)=-x\;\forall x,d\in\bbR\}\cap
\left\{E_n[Y_{n+1}]=\frac12\right\}
$$
it holds
$\frac12E_n[\beta_{n+1}]=E_n[Y_{n+1}\beta_{n+1}]=\frac12$,
i.e., $E_n[\beta_{n+1}]=1$. Hence,
$$
A_n
\subseteq
\left\{ E_n\left[ Y_{n+1}\right] = \frac12, E_n\left[\beta_{n+1}\right] =1\right\}
=
\left\{Y_n=\frac{1}{2}\right\},
$$
where the set equality is again
Proposition~\ref{prop:31052019a1}.
This concludes the proof.
\end{proof}

\begin{corollary}\label{cor:01062019a1}
Under the assumptions of Theorem~\ref{thm:mainres}
it holds
$$
\left\{Y_{N-1}=\frac12\right\}
=\{\xi^*_{N-1}(x,d)=-x\;\forall x,d\in\bbR\}.
$$
\end{corollary}

\begin{proof}
This follows from
part~\ref{it:01062019a2}
of Proposition~\ref{prop:01062019a1}
because $Y_N=\frac12$.
\end{proof}

We now provide more details for the case of processes
with independent multiplicative increments
of Section~\ref{sec:PIMI}.
We recall that in this case the process $Y$
is deterministic.
Notice, however, that the trades $\xi^*_n(x,d)$ are still,
in general, random because of the randomness
in~$\gamma_n$, see~\eqref{eq:13072019a1}.

\begin{propo}\label{rem:close_immed}
Let $n\in\bbZ\cap(-\infty,N-1]$.
Under the assumptions of Lemma~\ref{prop:mii}
it holds:
\begin{enumerate}
\item
$\{\xi^*_n(x,d)=-x\;\forall x,d\in\bbR\}$
is either $\Omega$ or $\emptyset$.

\item
The following statements are equivalent:

\smallskip
(i) $\{\xi^*_n(x,d)=-x\;\forall x,d\in\bbR\}=\Omega$.

\smallskip
(ii) There exist $x,d\in\R$ with
$P(\gamma_nx\neq d)>0$ such that
$\{\xi^*_n(x,d)=-x\}=\Omega$.

\smallskip
(iii) It holds that
\begin{equation}\label{eq:01062019a2}
E[\beta_{n+1}]=1+\frac{\left(1-E\left[\frac{\beta_{n+1}^2}{\eta_{n+1}}\right]\right)
\left(\frac{1}{2}-Y_{n+1}\right)}{Y_{n+1}}.
\end{equation}

\item\label{it:01062019a3}
Under~\eqref{eq:01062019a2}
we have that $E[\beta_{n+1}]\ge 1$
and, if $Y_{n+1}<\frac{1}{2}$,
even that $E[\beta_{n+1}]>1$.
\end{enumerate}
\end{propo}

The meaning of part~\ref{it:01062019a3} in Proposition~\ref{rem:close_immed} is that,
in the case of (PIMI)
(special case: deterministic processes $\beta$ and $\gamma$),
closing the position in one go
is never optimal in the (usual) framework,
where the resilience process $\beta$ is assumed to be $(0,1)$-valued.

This raises the question of whether closing the position
in one go can be optimal in general
(that is, beyond (PIMI))
with the resilience process $\beta$
taking values in $(0,1)$.
In our setting the answer is affirmative
(see Example~\ref{ex:13072019a1} below).
It is worth noting that in the related setting,
where trading is constrained only in one direction
and the process $\beta$ is $(0,1)$-valued,
the answer is negative,
i.e., closing the position in one go
is never optimal
(see Proposition~A.3 in \cite{fruth2019optimal}
and Proposition~5.6 in \cite{fruth2014optimal}).

\begin{proof}
1. Since $Y$ is deterministic
and $\eta_{n+1}$ and $\beta_{n+1}$
are independent of $\cF_n$,
Lemma~\ref{lem:01062019a1} yields
\begin{equation}\label{eq:01062019a3}
\{\xi^*_n(x,d)=-x\;\forall x,d\in\bbR\}=
\left\{
\left(Y_{n+1}-\frac12\right)
E\left[\frac{\beta_{n+1}^2}{\eta_{n+1}}\right]
-Y_{n+1}E\left[\beta_{n+1}\right]+\frac12
=0\right\},
\end{equation}
which can be either $\Omega$ or $\emptyset$.

\smallskip
2. The equivalence between (i) and~(ii)
is a direct calculation using~\eqref{eq:optstrat}
and the fact that the factor in front of
$(x-\frac d{\gamma_n})$
on the right-hand side of~\eqref{eq:optstrat}
is deterministic under our assumptions.
The equivalence between (i) and~(iii)
follows from~\eqref{eq:01062019a3}
via a straightforward calculation.

\smallskip
3. The last statement is clear.
\end{proof}

We close the section with two examples announced above.

\begin{ex}\label{ex:01062019a1}
Consider the processes $\beta$ and $\gamma$
satisfying the assumptions of Lemma~\ref{prop:mii}
(in particular, (PIMI)) and, moreover,
$E[\beta_N]\neq 1$ and
\begin{equation}\label{eq:close_immed}
E[\beta_{N-1}]=1+\frac{\left(1-E\left[\frac{\beta_{N-1}^2}{\eta_{N-1}}\right]\right)
\left(\frac{1}{2}-Y_{N-1}\right)}{Y_{N-1}}.
\end{equation}
Below we present a specific choice of the parameters
such that~\eqref{eq:close_immed} is satisfied.

As we are in the framework of (PIMI),
the process $Y$ is deterministic.
Moreover, since $E[\beta_N]\ne1$,
we have $Y_{N-1}\in(0,\frac12)$
(see Corollary~\ref{cor:30052019a1}).
Recall that
on $\{Y_{N-1}<\frac12\}$ ($=\Omega$, up to a $P$-null set)
there exist profitable round trips when we start
at time $N-1$ with $d\ne0$. In particular,
\begin{equation}\label{eq:01062019a4}
P(\xi^*_{N-1}(0,d)\ne0)=1\quad
\text{whenever }d\ne0.
\end{equation}
That is, even without an open position we trade at time $N-1$ as soon as $d\neq 0$.\footnote{For completeness
we mention the explicit formula
$$
\xi^*_{N-1}(0,d)
=\frac{E[\beta_N]-1}{E[\eta_N-2\beta_N+1]}\frac{d}{\gamma_{N-1}},
$$
which can be obtained from~\eqref{eq:optstrat}
via a direct calculation
and yields an alternative proof of~\eqref{eq:01062019a4}.}

Moreover, notice that by Proposition~\ref{rem:close_immed}
condition~\eqref{eq:close_immed} is equivalent to
\begin{equation}\label{eq:01062019a5}
\{\xi^*_{N-2}(x,d)=-x\;\forall x,d\in\bbR\}=\Omega.
\end{equation}
To summarize, the optimal strategy in this example is to close the position at time $N-2$, to build up a new position at time $N-1$ (at least if $D_{(N-1)-}=(d-\gamma_{N-2}x)\beta_{N-1})\neq 0$) and to close this position at time $N$. Interestingly, such a phenomenon can only occur if $E[\beta_{N-1}]>1$,
and hence it cannot happen in the (usual) framework, where the resilience process $\beta$ is assumed to take values in $(0,1)$.

We, finally, remark that in this example
the inclusion in part~\ref{it:01062019a1}
of Proposition~\ref{prop:01062019a1}
for time $n=N-2$ is strict
(cf.\ \eqref{eq:01062019a5} with the fact that
$\{Y_{N-2}=\frac12\}=\emptyset$,
where the latter follows from $Y_{N-1}<\frac12$
and part~\ref{propYitem2}
of Corollary~\ref{propo:propY}).\footnote{More generally,
the inclusion in part~\ref{it:01062019a1}
of Proposition~\ref{prop:01062019a1} is strict
whenever on a set of positive probability
we have the phenomenon described in the previous paragraph.
Indeed, an event, where such a phenomenon happens,
is a subset of
$\{\xi^*_n(x,d)=-x\;\forall x,d\in\bbR\}
\setminus\{Y_n=\frac12\}$
because on $\{Y_n=\frac12\}$
we have $Y_n=Y_{n+1}=\ldots=Y_{N-1}=\frac12$
(part~\ref{propYitem2} of Corollary~\ref{propo:propY})
and hence
$\xi^*_k(x,d)=-x$ for all $x,d\in\bbR$
and $k\in\{n,n+1,\ldots,N-1\}$
(part~\ref{it:01062019a1} of Proposition~\ref{prop:01062019a1}),
in particular, $\xi^*_k(0,d)=0$ for all such $k$ and $d\in\bbR$.}

\smallskip
It remains to explain how we can satisfy~\eqref{eq:close_immed}. 
An easy specific example, where the requirements
on $\beta$ and $\gamma$
listed 
above are satisfied,
can be constructed with deterministic sequences
$\beta$ and $\gamma$.
For instance, choose arbitrary deterministic
$\gamma_N,\gamma_{N-1}>0$
and $\beta_N\in(0,\sqrt{\eta_N})\setminus\{1\}$.
These inputs yield a deterministic
$Y_{N-1}\in(0,\frac12)$
(see Corollary~\ref{cor:30052019a1}).
Take a sufficiently small $a>0$ such that
$$
\frac{aY_{N-1}}{\frac12-Y_{N-1}}\in(0,1).
$$
Finally, set $\beta_{N-1}=1+a$ and
choose $\gamma_{N-2}>0$ to satisfy
$$
\frac{aY_{N-1}}{\frac12-Y_{N-1}}
=1-\frac{(1+a)^2}{\eta_{N-1}}
$$
(recall that $\eta_{N-1}=\frac{\gamma_{N-1}}{\gamma_{N-2}}$).
This choice gives us~\eqref{eq:close_immed}
together with $\frac{\beta_{N-1}^2}{\eta_{N-1}}<1$.
\end{ex}

\begin{ex}\label{ex:13072019a1}
In this example we consider a version of our model with three trading periods $N-2$, $N-1$ and $N$, where the resilience process $\beta$ is $(0,1)$-valued and still it is optimal at time $N-2$ to close the position in one go. To this end assume that $\mathcal{F}_{N-2} = \{\emptyset, \Omega\}$ and $\mathcal{F}_{N-1} = \sigma\left( \gamma_{N-1} \right)$ and that we can specify the positive random variables
$\gamma_{N-1}$, $\gamma_N$ and the $(0,1)$-valued
random variable
$\beta_N$ in such a way that
$E_{N-1}\left[\frac{\beta_N^2}{\eta_N} \right] < 1$, 
$\left(1-E_{N-1}\left[\frac{\beta_N^2}{\eta_N} \right]\right)^{-1}
\in L^{\infty-}$ 
and that $Y_{N-1}$ and $\frac{1}{\gamma_{N-1}}$ are strictly negatively correlated, i.e.,
\begin{equation}\label{eq:strictlynegativelycorr}
E\left[ \frac{Y_{N-1}}{\gamma_{N-1}} \right] - E\left[ Y_{N-1} \right] E\left[ \frac{1}{\gamma_{N-1}} \right] < 0 .
\end{equation}
Below we present a specific choice such that these assumptions are satisfied.
By~\eqref{eq:strictlynegativelycorr} we can choose
a deterministic
\begin{equation}\label{eq:def_beta_ex}
\beta_{N-1} \in \left( \frac{E\left[ \frac{Y_{N-1}}{\gamma_{N-1}} \right]}{E\left[ Y_{N-1} \right] E\left[ \frac{1}{\gamma_{N-1}} \right]} , 1 \right)
\end{equation}
and then define
\begin{equation}\label{eq:def_gamma_ex}
\gamma_{N-2}=\frac{E\left[ \frac{1}{2} - Y_{N-1} \beta_{N-1} \right]}{E\left[ \left( \frac{1}{2} - Y_{N-1} \right) \frac{\beta_{N-1}^2}{\gamma_{N-1}} \right]}. 
\end{equation}
Note that, indeed, $\beta_{N-1} \in (0,1)$ and
$\gamma_{N-2}>0$.
Next, we verify that $E\left[ \frac{\beta_{N-1}^2}{\eta_{N-1}} \right] < 1$. By~\eqref{eq:def_beta_ex} it holds 
$
E\left[ \beta_{N-1} Y_{N-1} \right] E\left[ \frac{1}{\gamma_{N-1}} \right] > E\left[ \frac{Y_{N-1}}{\gamma_{N-1}} \right].
$
This implies 
\begin{equation*}
E\left[ \frac{1}{2} - \beta_{N-1} Y_{N-1} \right] E\left[ \frac{\beta_{N-1}^2}{\gamma_{N-1}} \right]
< E\left[ \left( \frac{1}{2} - Y_{N-1}\right) \frac{\beta_{N-1}^2}{\gamma_{N-1}}\right]
\end{equation*}
and hence
\begin{equation*}
\gamma_{N-2} = \frac{E\left[ \frac{1}{2} - \beta_{N-1} Y_{N-1} \right]}{ E\left[ \left( \frac{1}{2} - Y_{N-1}\right) \frac{\beta_{N-1}^2}{\gamma_{N-1}}\right]} 
< \frac{1}{ E\left[ \frac{\beta_{N-1}^2}{\gamma_{N-1}} \right]}.
\end{equation*}
Since $\gamma_{N-2}$ is deterministic and $\eta_{N-1}=\frac{\gamma_{N-1}}{\gamma_{N-2}}$, we get $E\left[ \frac{\beta_{N-1}^2}{\eta_{N-1}} \right] < 1$. 

From~\eqref{eq:def_gamma_ex} we obtain that 
\begin{equation*}
E\left[ \left(Y_{N-1}-\frac12\right)\frac{\beta_{N-1}^2}{\eta_{N-1}}
-Y_{N-1}\beta_{N-1}+\frac12  \right] = 0 .
\end{equation*}
Therefore, it follows from Lemma~\ref{lem:01062019a1} that for all $x,d\in\bbR$ it holds that
$\xi^*_{N-2}(x,d)=-x$, i.e., it is optimal to close the whole position at time $N-2$.

\smallskip
It remains to specify $\gamma_{N-1}$, $\gamma_N$ and $\beta_N$ such that
$E_{N-1}\left[\frac{\beta_N^2}{\eta_N} \right] < 1$, 
$\left(1-E_{N-1}\left[\frac{\beta_N^2}{\eta_N} \right]\right)^{-1}
\in L^{\infty-}$ 
and that \eqref{eq:strictlynegativelycorr} is satisfied. To this end let $\gamma_{N-1}$ be $\{ \frac{1}{2}, 1 \}$-valued with $P\left( \gamma_{N-1} = 1 \right) = p \in (0,1)$ and  $P\left( \gamma_{N-1} = \frac{1}{2} \right) = 1- p$. Define $\gamma_N = \gamma_{N-1}^2$ and 
$\beta_N=\frac{\gamma_{N-1}}2$.

Note that $\beta_N$ is $(0,1)$-valued,
$\gamma_{N-1},\gamma_N>0$ and $\eta_N = \gamma_{N-1}$. Observe further that $E_{N-1}\left[ \beta_{N}^2 \right] =
\frac{\gamma_{N-1}^2}4
< \gamma_{N-1}$ and hence 
$E_{N-1}\left[\frac{\beta_N^2}{\eta_N} \right]
=\frac{\gamma_{N-1}}4\le\frac14< 1$
and
$\left(1-E_{N-1}\left[\frac{\beta_N^2}{\eta_N} \right]\right)^{-1}
\in L^{\infty-}$. 
By definition of $\beta_{N-1}$, we have $E_{N-1}\left[ \beta_N \right] = \frac{\gamma_{N-1}}{2}$. 
It therefore holds 
\begin{equation*}
	Y_{N-1} = \frac{1}{2} \,\frac{ E_{N-1}\left[ \eta_N \right] - \left( E_{N-1}\left[ \beta_N \right]  \right)^2 }{ 1 - 2 E_{N-1}\left[ \beta_N \right] + E_{N-1}\left[ \eta_N \right] }
= \frac{1}{2} \left( \gamma_{N-1} - \frac{\gamma_{N-1}^2}{4} \right).
\end{equation*}
Since 
\begin{equation*}
E\left[ \gamma_{N-1}\right] = p + \frac{1}{2} (1-p), \, 
E\left[ \gamma_{N-1}^2 \right] = p + \frac{1}{4} (1-p) \text{ and }
E\left[ \frac{1}{\gamma_{N-1}} \right] = p + 2(1-p),
\end{equation*}
we obtain~\eqref{eq:strictlynegativelycorr}:
\begin{equation*}
\begin{split}
& E\left[ \frac{Y_{N-1}}{\gamma_{N-1}} \right] - E\left[ Y_{N-1} \right] E\left[ \frac{1}{\gamma_{N-1}} \right] \\
& = \frac{1}{2}\left( 1-\frac{1}{4} E\left[ \gamma_{N-1} \right] \right) - \frac{1}{2} \left( E\left[ \gamma_{N-1} \right] - \frac{1}{4} E\left[ \gamma_{N-1}^2 \right] \right) E\left[ \frac{1}{\gamma_{N-1}} \right] =  \frac{1}{2} \frac{5}{16} \left( p^2 - p\right) < 0.
\end{split}
\end{equation*}
For completeness, we notice
that the assumptions of Theorem~\ref{thm:mainres}
which were not explicitly discussed above
(e.g., $\beta_n,\gamma_n,\frac1{\gamma_n}\in L^{\infty-}$)
are trivially satisfied. 
\end{ex}

Finally, Table~\ref{tab:09062020a1} summarizes several mentioned qualitative effects and compares our findings with the literature.
In this table, the term ``one-directional trading'' refers to settings, where the trading is constrained in one direction,
and the term ``two-directional trading'' refers to settings, where, like in the present paper, trading in both directions is allowed.

\newlength{\myscale}
\newlength{\tinywidth}
\newlength{\smallwidth}
\newlength{\largewidth}
\setlength{\myscale}{0.92\textwidth}
\setlength{\tinywidth}{0.033\myscale}
\setlength{\largewidth}{0.2\myscale}
\setlength{\smallwidth}{0.333\myscale}
\addtolength{\smallwidth}{-\largewidth}

\begin{table}[!htb]
\caption{We compare different settings from the viewpoint of whether premature closure is possible. Columns 2--4 briefly describe the settings, while columns 5--6 provide the answers and references to the proofs. It is worth noting that setting~1 is studied in \cite{bank2014optimal} and \cite{fruth2014optimal}, setting~2 in \cite{fruth2019optimal} and setting~3 in \cite{alfonsi2014optimal} (although the question of closing the position in one go is not explicitly considered in \cite{alfonsi2014optimal}, hence the reference to our paper in the last column).}\label{tab:09062020a1}
\begin{tabular}{|p{\tinywidth}|p{\largewidth}|p{\largewidth}|p{\smallwidth}|p{\smallwidth}|p{\largewidth}|}
\hline
&
One- or two-directional trading?&
$\beta$ and $\gamma$ deterministic or stochastic?&
$\beta$ $(0,1)$-valued?&
Premature closure possible?&
Reason
\\
\hline
1 & one-directional & deterministic & yes & no & Proposition~5.6 in \cite{fruth2014optimal}
\\
\hline
2 & one-directional & stochastic & yes & no & Proposition~A.3 in \cite{fruth2019optimal}
\\
\hline
3 & two-directional & deterministic & yes & no & Proposition~\ref{rem:close_immed} in this paper
\\
\hline
4 & two-directional & deterministic & no & yes & Example~\ref{ex:01062019a1} in this paper
\\
\hline
5 & two-directional & stochastic & yes & yes & Example~\ref{ex:13072019a1} in this paper
\\
\hline
6 & two-directional & stochastic & no & yes & trivial (follows from 4 or~5)
\\
\hline
\end{tabular}
\end{table}

\appendix
\section{Proof of Theorem~\ref{thm:mainres}}
\label{a:proof_mainres}
\begin{proof}
We first prove \eqref{eq:sol_valfct} and~\eqref{eq:optstrat}
by backward induction on $n \in \Z \cap (-\infty,N]$. For the base case $n=N$ note that for all $x,d \in \R$ it holds that $V_N(x,d)=-(d-\frac{\gamma_N}{2}x)x=\frac{\gamma_N}{2}\left(\frac{d}{\gamma_N}-x\right)^2-\frac{d^2}{2\gamma_N}$. In particular, it holds that $Y_N=\frac{1}{2}>0$. Besides that, we have that for all $x,d\in \R$, $\xi_N^*(x,d)=-x$ is the unique element of $\mathcal A_N(x)$ and hence optimal.

Consider now the induction step $\Z \cap (-\infty,N] \ni n+1\to n \in \Z \cap (-\infty,N-1]$. For all $x,d \in \R$ let 
\begin{equation}\label{eq:abc}
\begin{split}
a_n&=\gamma_n E_n\left[\frac{Y_{n+1}}{\eta_{n+1}}\left(\beta_{n+1}-\eta_{n+1}\right)^2 +\frac{1}{2}\left(1-\frac{\beta_{n+1}^2}{\eta_{n+1}}\right)\right],\\
b_n(x,d)&=E_n\left[d\left(1-\frac{\beta_{n+1}^2}{\eta_{n+1}}\right)+2Y_{n+1}\left(\frac{\beta_{n+1}}{\eta_{n+1}}-1\right)\left(\beta_{n+1}d-\gamma_{n+1}x\right) \right],\\
c_n(x,d)&=E_n\left[\frac{Y_{n+1}}{\gamma_{n+1}}\left(\beta_{n+1}d-\gamma_{n+1}x\right)^2 -\frac{d^2\beta_{n+1}^2}{2\gamma_{n+1}}\right].
\end{split}
\end{equation}
Note that for all $x,d\in\R$ the random variables $a_n, b_n(x,d)$ and $c_n(x,d)$ are well-defined and finite because all factors and summands are in $\Linfm$ due to the assumption that for all $k\in \Z\cap(-\infty,N]$, it holds $\beta_k, \gamma_k, \frac{1}{\gamma_k} \in \Linfm$, and the induction hypothesis $0<Y_{n+1}\leq\frac{1}{2}$. Furthermore, the induction hypothesis that $Y_{n+1}>0$ and the assumption that $E_n\left[\frac{\beta^2_{n+1}}{\eta_{n+1}}\right]<1$ ensure that $a_n>0$. It follows from the Cauchy-Schwarz inequality and the assumption $E_n\left[\frac{\beta^2_{n+1}}{\eta_{n+1}}\right]<1$ that
\begin{equation}
\begin{split}
Y_n&=E_n[\eta_{n+1}Y_{n+1}]-\frac{\left(E_n\left[\sqrt{\eta_{n+1}Y_{n+1}}\sqrt{\frac{Y_{n+1}}{\eta_{n+1}}}\left(\beta_{n+1}-\eta_{n+1}\right) \right]\right)^2}{a_n/\gamma_n}\\
&\ge E_n[\eta_{n+1}Y_{n+1}]-\frac{E_n\left[\eta_{n+1}Y_{n+1}\right]E_n\left[\frac{Y_{n+1}}{\eta_{n+1}}\left(\beta_{n+1}-\eta_{n+1}\right)^2 \right]}{a_n/\gamma_n}\\
&=\frac{E_n[\eta_{n+1}Y_{n+1}]}{a_n/\gamma_n}\,
\frac{1}{2}
\left(1-E_n\left[\frac{\beta^2_{n+1}}{\eta_{n+1}}\right]\right)>0.
\end{split}
\end{equation}
To establish that $Y_n\leq\frac{1}{2}$, note that  
\begin{equation}
\begin{split}
\frac{1}{\gamma_n}c_n(1,0) & = \frac{1}{\gamma_n} E_n\left[ Y_{n+1}\gamma_{n+1} \right] = E_n\left[ \eta_{n+1} Y_{n+1} \right], \\
\frac{1}{\gamma_n} b_n(1,0) & = \frac{1}{\gamma_n} E_n\left[ -2 Y_{n+1}\left( \frac{\beta_{n+1}}{\eta_{n+1}} - 1\right) \gamma_{n+1} \right]  = -2  E_n\left[ Y_{n+1}\left( \beta_{n+1}- \eta_{n+1} \right) \right].
\end{split}
\end{equation}
This together with the induction hypothesis $Y_{n+1}\leq \frac{1}{2}$ implies that 
\begin{equation}\label{eq:Yupb}
\begin{split}
Y_n 
& = \frac{1}{\gamma_n} \left( c_n(1,0) - \frac{ b_n(1,0)^2 }{4 a_n} \right) 
 \leq \frac{1}{\gamma_n} \left( c_n(1,0) - \frac{ b_n(1,0)^2 }{4 a_n} + a_n \left( \frac{b_n(1,0)}{2 a_n} - 1 \right)^2 \right) \\
& = \frac{1}{\gamma_n} \left( a_n - b_n(1,0) + c_n(1,0) \right) 
= \frac{1}{2} + E_n\left[ \frac{\beta_{n+1}^2}{\eta_{n+1}} \left( Y_{n+1} - \frac{1}{2} \right) \right] 
 \leq \frac{1}{2}.
\end{split}
\end{equation}
Let $\mathcal S_n$ be the set of all real-valued $\mathcal F_n$-measurable random variables $\xi \in \Ltwop$. The dynamic programming principle
and the induction hypothesis ensure for all $x,d \in \R$ that\footnote{Note that our assumption that for all $k\in \Z\cap(-\infty,N]$ it holds $\beta_k, \gamma_k, \frac{1}{\gamma_k} \in \Linfm$  and the fact that $Y_{n+1}$ is bounded ensure that all conditional expectations in \eqref{eq:gen_dpp} are well-defined and that we can move any $\xi\in \cS_n$, $\gamma_n$ and $\frac{1}{\gamma_n}$ outside the conditional expectations. This reasoning also applies to other calculations in this proof.}
\begin{equation}\label{eq:gen_dpp}
\begin{split}
& V_n(x,d)\\
&=\essinf_{\xi \in \mathcal S_n}\left[  \left(d+\frac{\gamma_n}{2}\xi\right)\xi + E_n\left[V_{n+1}(x+\xi,(d+\gamma_{n}\xi)\beta_{n+1} )\right]\right]\\
&=\essinf_{\xi \in \mathcal S_n}\left[  \left(d+\frac{\gamma_n}{2}\xi\right)\xi + E_n\left[\frac{Y_{n+1}}{\gamma_{n+1}}\left((d+\gamma_{n}\xi)\beta_{n+1}-\gamma_{n+1}(x+\xi)\right)^2-\frac{(d+\gamma_{n}\xi)^2\beta_{n+1}^2}{2\gamma_{n+1}} \right]\right]\\
&=\essinf_{\xi \in \mathcal S_n}\left[  \left(d+\frac{\gamma_n}{2}\xi\right)\xi + E_n\left[\gamma_{n+1}Y_{n+1}\left(\left(\frac{\beta_{n+1}\gamma_{n}}{\gamma_{n+1}}-1\right)\xi+\frac{d\beta_{n+1}}{\gamma_{n+1}}-x\right)^2-\frac{(d+\gamma_{n}\xi)^2\beta_{n+1}^2}{2\gamma_{n+1}} \right]\right]\\
&=\essinf_{\xi \in \mathcal S_n}\left[a_n\xi^2+b_n(x,d)\xi+c_n(x,d)\right].
\end{split}
\end{equation}
For all $x,d \in \R$ we find $\xi^*_n(x,d)=-\frac{b_n(x,d)}{2a_n}$ to be the unique minimizer of $\xi \mapsto a_n\xi^2+b_n(x,d)\xi+c_n(x,d)$. Observe further that for all $x,d\in \R$ it holds that 
\begin{equation}
b_n(x,d)=\frac{2da_n}{\gamma_n}-2E_n\left[Y_{n+1}\left(\beta_{n+1}\gamma_n-\gamma_{n+1}\right) \right]\left(x-\frac{d}{\gamma_n}\right),
\end{equation}
which yields the representation of $\xi^*_n(x,d)$ in~\eqref{eq:optstrat}.
Clearly, for all $x,d\in\bbR$ the random variable
$\xi^*_n(x,d)$ is $\cF_n$-measurable.
It remains to verify that for all $x,d\in\bbR$
we have $\xi^*_n(x,d)\in\Linfm$.

To show this we verify first that 
\begin{equation}\label{eq:ratioinallLp}
\frac{ E_n\left[ Y_{n+1}\left(\beta_{n+1}-\eta_{n+1}\right) \right]  }{E_n\left[ \frac{Y_{n+1}}{\eta_{n+1}}\left(\beta_{n+1}-\eta_{n+1}\right)^2+\frac{1}{2}\left(1-\frac{\beta_{n+1}^2}{\eta_{n+1}}\right)  \right]} \in \Linfm.
\end{equation}
We have $\eta_{n+1} \in \Linfm$ as $\eta_{n+1}$ is the product of the two $\Linfm$-variables $\gamma_{n+1}$ and $\frac{1}{\gamma_{n}}$. Furthermore, we have that $\beta_{n+1} \in \Linfm$ by assumption and that $Y_{n+1}$ is bounded due to the induction hypothesis. Hence, by the Minkowski inequality, it holds that 
\begin{equation}
\left(E\left[ \lvert Y_{n+1}\left(\beta_{n+1}-\eta_{n+1}\right) \rvert^p \right] \right)^\frac{1}{p}\leq \left(E\left[ \lvert Y_{n+1}\beta_{n+1} \rvert^p \right] \right)^\frac{1}{p} + \left(E\left[ \lvert Y_{n+1}\eta_{n+1} \rvert^p \right] \right)^\frac{1}{p} < \infty
\end{equation} 
for every $p\in [1,\infty)$, so that 
\begin{equation}\label{eq:nominallLp}
E_n\left[ Y_{n+1}\left(\beta_{n+1}-\eta_{n+1}\right) \right] \in \Linfm.
\end{equation} 
Next we recall that $\frac1{\alpha_n}\in\Linfm$,
where $\alpha_n=1-E_n\left[\frac{\beta^2_{n+1}}{\eta_{n+1}}\right]$,
which implies
\begin{equation}\label{eq:afromY}
\frac{1}{E_n\left[ \frac{Y_{n+1}}{\eta_{n+1}}\left(\beta_{n+1}-\eta_{n+1}\right)^2+\frac{1}{2}\left(1-\frac{\beta_{n+1}^2}{\eta_{n+1}}\right)  \right]} \in \Linfm,
\end{equation}
as the random variable in~\eqref{eq:afromY}
is positive and smaller than $\frac2{\alpha_n}$.
Together with~\eqref{eq:nominallLp} this establishes~\eqref{eq:ratioinallLp}.
Now \eqref{eq:optstrat} and~\eqref{eq:ratioinallLp}
imply that $\xi^*_n(x,d)\in\Linfm$ for all $x,d\in\bbR$,
as $x$ and $d$ are deterministic and $\frac{1}{\gamma_n}\in\Linfm$.

By inserting the optimal trade size $\xi^*_n(x,d)=-\frac{b_n(x,d)}{2a_n}$ into \eqref{eq:gen_dpp}, we obtain for all $x,d \in \R$ that
\begin{equation}\label{eq:quadformV}
V_n(x,d)=-\frac{b_n(x,d)^2}{4a_n}+c_n(x,d).
\end{equation}
The dynamic programming principle ensures for all $x,d,h \in \R$ that
\begin{equation}\label{eq:prep_deriv_V}
\begin{split}
&V_n(x,d)-\left(d+\frac{\gamma_n}{2}h\right)h\\
&=\essinf_{\xi \in \mathcal S_n}\left[ \left(d+\frac{\gamma_n}{2}\xi\right)\xi-\left(d+\frac{\gamma_n}{2}h\right)h + E_n\left[V_{n+1}(x+\xi,(d+\gamma_{n}\xi)\beta_{n+1} )\right]\right]\\
&=\essinf_{ \xi \in \mathcal S_n}\left[  \left(d+\frac{\gamma_n}{2}(\xi+h)\right)( \xi-h) + E_n\left[V_{n+1}(x+ \xi,(d+\gamma_n\xi)\beta_{n+1} )\right]\right]\\
&=\essinf_{\tilde \xi \in \mathcal S_n}\left[  \left(d+\gamma_nh+\frac{\gamma_n}{2}\tilde \xi\right)\tilde \xi + E_n\left[V_{n+1}(x+h+\tilde \xi,(d+\gamma_n(h+\tilde \xi))\beta_{n+1} )\right]\right]\\
&=V_n(x+h,d+\gamma_nh).
\end{split}
\end{equation}
Note that by \eqref{eq:quadformV} and~\eqref{eq:abc} it holds that for almost all $\omega$, $V_n$ is a quadratic function in $(x,d)\in \R^2$ with $V_n(0,0)=0$. In particular, derivatives in the argumentation below exist.  Equation~\eqref{eq:prep_deriv_V} 
implies for all $x,d\in \R$ that
\begin{equation}
(\partial_x V_n)(x,d)+\gamma_n(\partial_d V_n)(x,d)\leftarrow \frac{V_n(x+h,d+\gamma_nh)-V_n(x,d)}{h}=-\left(d+\frac{\gamma_n}{2}h\right)\to -d
\end{equation}
as $h\to 0$. 
Consequently, 
we obtain that 
\begin{equation}\label{eq:secderV}
(\partial^2_{xx}V_n)(0,0)+\gamma_n (\partial^2_{dx}V_n)(0,0)=0\quad \text{and} \quad (\partial^2_{xd}V_n)(0,0)+\gamma_n(\partial^2_{dd}V_n)(0,0)=-1.
\end{equation} 
This, together with \eqref{eq:quadformV} and~\eqref{eq:abc}, proves that 
\begin{equation}\label{eq:quadfromV2}
\begin{split}
V_n(x,d)&=\frac{(\partial^2_{xx}V_n)(0,0)}{2}x^2+[(\partial^2_{dx}V_n)(0,0)]xd+\frac{(\partial^2_{dd}V_n)(0,0)}{2}d^2\\
&=\frac{(\partial^2_{xx}V_n)(0,0)}{2}\left(\frac{d}{\gamma_n}-x\right)^2-\frac{d^2}{2\gamma_n}.
\end{split}
\end{equation}
Moreover, it follows from \eqref{eq:quadformV} that
\begin{equation}
\frac{(\partial^2_{xx}V_n)(0,0)}{2}=E_n[\gamma_{n+1}Y_{n+1}]-\frac{\left(E_n\left[Y_{n+1}\left(\beta_{n+1}\gamma_n-\gamma_{n+1}\right) \right]\right)^2}{a_n}=\gamma_nY_n.
\end{equation}
This together with \eqref{eq:quadfromV2} proves that $V_n(x,d)=\frac{Y_n}{\gamma_n}\left(d-x\gamma_n\right)^2-\frac{d^2}{2\gamma_n}$ for all $x,d\in \R$.

In the remainder of the proof we show that for all $n\in \Z \cap (-\infty,N-1]$, $x,d\in\R$ the process $\xi^*=\left(\xi_k^* \right)_{k\in\{n,\ldots,N\}}$ recursively defined by \eqref{eq:optprocess} is in $\mathcal A_n(x)$. To this end we show by (forward) induction on $k \in \{n,\ldots,N\}$ that $\xi_k^*$ is $\mathcal{F}_k$-measurable and belongs to $\Ltwop$ for all $k \in \{n,\ldots,N\}$.

For the base case $k=n$ we have $\xi_n^*=\xi_n^*(x,d)$ which is already known to be in $\mathcal S_n$ for all $x,d \in \R$, i.e., $\xi_n^*$ is $\mathcal F_n$-measurable and $\xi_n^* \in \Ltwop$.

Continue with the induction step $\{n,\ldots,N-2\} \ni k-1 \to k \in \{n+1,\ldots,N-1\}$. Now, the optimal trade size $\xi_k^*$ at time $k$ depends on the current value of the position path $X^*_{k-1}=x+\sum_{i=n}^{k-1} \xi_i^*$ and the current deviation $D^*_{k-}$. By induction on $k$, it holds that $\xi_i^*$ is in $\Ltwop$ and $\mathcal F_i$-measurable for all $i \in \{n,\ldots,k-1\}$. This yields that $X^*_{k-1}$ belongs to $\Ltwop$ and is $\mathcal F_{k}$-measurable. Furthermore, the fact that $\xi_i^* \in \Ltwop$  for all $i \in \{n,\ldots,k-1\}$ allows us to use Remark \ref{rem:DinL2+} to obtain that $D^*_{k-} \in \Ltwop$ as well. Besides that, it can be seen from \eqref{eq:formD} that $D^*_{k-}$ is $\mathcal F_k$-measurable given that $\xi_i^*$ is $\mathcal F_i$-measurable for all $i \in \{n,\ldots,k-1\}$ and $\beta$ and $\gamma$ are adapted processes. Hence, 
\begin{equation}\label{eq:xijXD}
\xi^*_k(X^*_{k-1},D^*_{k-})=\frac{E_k\left[Y_{k+1}\left(\beta_{k+1}-\eta_{k+1}\right) \right]}{E_k\left[\frac{Y_{k+1}}{\eta_{k+1}}\left(\beta_{k+1}-\eta_{k+1}\right)^2+\frac{1}{2}\left(1-\frac{\beta_{k+1}^2}{\eta_{k+1}}\right) \right]}\left(X^*_{k-1}-\frac{D^*_{k-}}{\gamma_k}\right)-\frac{D^*_{k-}}{\gamma_k} 
\end{equation}
is $\mathcal F_k$-measurable. To prove that $\xi_k^*(X^*_{k-1},D^*_{k-}) \in \Ltwop$, note that by the Minkowski inequality, it suffices to show that each summand is in $\Ltwop$. To begin with, it holds that $\frac{D^*_{k-}}{\gamma_k} \in \Ltwop$ due to Lemma \ref{lem:inL2+} and $\frac{1}{\gamma_{k}}\in \Linfm$. It further follows with \eqref{eq:ratioinallLp} and Lemma \ref{lem:inL2+} that 
\begin{equation}
\frac{E_k\left[Y_{k+1}\left(\beta_{k+1}-\eta_{k+1}\right) \right]}{E_k\left[\frac{Y_{k+1}}{\eta_{k+1}}\left(\beta_{k+1}-\eta_{k+1}\right)^2+\frac{1}{2}\left(1-\frac{\beta_{k+1}^2}{\eta_{k+1}}\right) \right]} \frac{D^*_{k-}}{\gamma_k} \in \Ltwop.
\end{equation}
Similarly,
\begin{equation}
\frac{E_k\left[Y_{k+1}\left(\beta_{k+1}-\eta_{k+1}\right) \right]}{E_k\left[\frac{Y_{k+1}}{\eta_{k+1}}\left(\beta_{k+1}-\eta_{k+1}\right)^2+\frac{1}{2}\left(1-\frac{\beta_{k+1}^2}{\eta_{k+1}}\right) \right]} X^*_{k-1} \in \Ltwop.
\end{equation}
This finishes the induction step $\{n,\ldots,N-2\} \ni k-1 \to k \in \{n+1,\ldots,N-1\}$.

Finally, it follows that for all $x,d \in \R$ it also holds true that $\xi_N^*=-X^*_{N-1}=-x-\sum_{i=n}^{N-1} \xi_i^*$ is in $\Ltwop$ and $\mathcal F_N$-measurable. As a result, $\xi^* \in \mathcal A_n(x)$ for all $x,d\in \R$.

The proof of Theorem~\ref{thm:mainres} is thus completed.
\end{proof}

\section{Integrability}\label{b:integrability}

\begin{lemma}
Let $X,Y \in \Linfm$. Then, $XY$ also belongs to $\Linfm$.
\end{lemma}

\begin{proof}
Let $p\in[1,\infty)$. The Cauchy-Schwarz inequality yields 
\begin{equation}
E\left[ \lvert XY \rvert^p \right] = E\left[ \lvert X\rvert^p \lvert Y\rvert^p \right] \leq \left( E\left[ \lvert X\rvert^{2p} \right] \right)^\frac{1}{2} \cdot \left( E\left[ \lvert Y\rvert^{2p} \right] \right)^\frac{1}{2} < \infty
\end{equation}
since $X, Y \in L^{2p}$. Therefore, $XY \in L^p$. This is true for every $p \in [1,\infty)$, hence $XY \in \Linfm$.
\end{proof}

\begin{lemma}\label{lem:inL2+}
Let $X \in \Linfm$ and $Y\in \Ltwop$. Then, $X Y \in \Ltwop$.
\end{lemma}

\begin{proof}
Since $Y \in \Ltwop$, there exists $\varepsilon>0$ such that $Y \in L^{2+\varepsilon}$. 
Let $r:=2+\frac{\varepsilon}{2}$ and $q:= \frac{2+\varepsilon}{r}$. It holds that $q>1$ and $Y\in L^{rq}$. Define $p:=\frac{q}{q-1}$ and observe that $X \in L^{rp}$. By the H{\"o}lder inequality,
\begin{equation}
E\left[ \lvert XY \rvert^r \right] = E\left[ \lvert X\rvert^r \lvert Y\rvert^r \right] \leq 
	\left( E\left[ \lvert X\rvert^{rp} \right] \right)^\frac{1}{p} \cdot \left( E\left[ \lvert Y\rvert^{rq} \right] \right)^\frac{1}{q}<\infty.
\end{equation}
This proves that $XY \in \Ltwop$.
\end{proof}

\bigskip\noindent
\textbf{Acknowledgement:}
We thank the Associate Editor and two anonymous referees for their constructive comments and suggestions that helped us improve the manuscript.

\bibliographystyle{abbrv}
\bibliography{literature}

\begin{thebibliography}{10}

\bibitem{ackermann2020cadlag}
J.~Ackermann, T.~Kruse, and M.~Urusov.
\newblock C\`adl\`ag semimartingale strategies for optimal trade execution in
  stochastic order book models.
\newblock {\em Preprint, arXiv:2006.05863}, 2020.

\bibitem{alfonsi2014optimal}
A.~Alfonsi and J.~I. Acevedo.
\newblock Optimal execution and price manipulations in time-varying limit order
  books.
\newblock {\em Applied Mathematical Finance}, 21(3):201--237, 2014.

\bibitem{alfonsi2008constrained}
A.~Alfonsi, A.~Fruth, and A.~Schied.
\newblock Constrained portfolio liquidation in a limit order book model.
\newblock {\em Banach Center Publ}, 83:9--25, 2008.

\bibitem{alfonsi2010optimal}
A.~Alfonsi, A.~Fruth, and A.~Schied.
\newblock Optimal execution strategies in limit order books with general shape
  functions.
\newblock {\em Quantitative Finance}, 10(2):143--157, 2010.

\bibitem{alfonsi2010boptimal}
A.~Alfonsi and A.~Schied.
\newblock Optimal trade execution and absence of price manipulations in limit
  order book models.
\newblock {\em SIAM Journal on Financial Mathematics}, 1(1):490--522, 2010.

\bibitem{alfonsi2012order}
A.~Alfonsi, A.~Schied, and A.~Slynko.
\newblock Order book resilience, price manipulation, and the positive portfolio
  problem.
\newblock {\em SIAM Journal on Financial Mathematics}, 3(1):511--533, 2012.

\bibitem{almgren2003optimal}
R.~Almgren.
\newblock Optimal execution with nonlinear impact functions and
  trading-enhanced risk.
\newblock {\em Applied mathematical finance}, 10(1):1--18, 2003.

\bibitem{almgren2012optimal}
R.~Almgren.
\newblock Optimal trading with stochastic liquidity and volatility.
\newblock {\em SIAM Journal on Financial Mathematics}, 3(1):163--181, 2012.

\bibitem{almgren1999value}
R.~Almgren and N.~Chriss.
\newblock Value under liquidation.
\newblock {\em Risk}, 12(12):61--63, 1999.

\bibitem{almgren2001optimal}
R.~Almgren and N.~Chriss.
\newblock Optimal execution of portfolio transactions.
\newblock {\em Journal of Risk}, 3:5--40, 2001.

\bibitem{ankirchner2020optimal}
S.~Ankirchner, A.~Fromm, T.~Kruse, and A.~Popier.
\newblock Optimal position targeting via decoupling fields.
\newblock {\em To appear in Annals of Applied Probability}, 2020.

\bibitem{ankirchner2014bsdes}
S.~Ankirchner, M.~Jeanblanc, and T.~Kruse.
\newblock {BSDE}s with singular terminal condition and a control problem with
  constraints.
\newblock {\em SIAM Journal on Control and Optimization}, 52(2):893--913, 2014.

\bibitem{ankirchner2015optimal}
S.~Ankirchner and T.~Kruse.
\newblock Optimal position targeting with stochastic linear-quadratic costs.
\newblock {\em Advances in Mathematics of Finance}, 104:9--24, 2015.

\bibitem{BD_AAP_2019}
P.~Bank and Y.~Dolinsky.
\newblock Continuous-time duality for superreplication with transient price
  impact.
\newblock {\em Ann. Appl. Probab.}, 29(6):3893--3917, 2019.

\bibitem{BD_Bern_2020}
P.~Bank and Y.~Dolinsky.
\newblock Scaling limits for super-replication with transient price impact.
\newblock {\em Bernoulli}, 26(3):2176--2201, 2020.

\bibitem{bank2014optimal}
P.~Bank and A.~Fruth.
\newblock Optimal order scheduling for deterministic liquidity patterns.
\newblock {\em SIAM Journal on Financial Mathematics}, 5(1):137--152, 2014.

\bibitem{bank2017hedging}
P.~Bank, H.~M. Soner, and M.~Vo{\ss}.
\newblock Hedging with temporary price impact.
\newblock {\em Mathematics and Financial Economics}, 11(2):215--239, 2017.

\bibitem{bank2018linear}
P.~Bank and M.~Vo{\ss}.
\newblock Linear quadratic stochastic control problems with stochastic terminal
  constraint.
\newblock {\em SIAM Journal on Control and Optimization}, 56(2):672--699, 2018.

\bibitem{BV_SIFIN_2019}
P.~Bank and M.~Vo\ss.
\newblock Optimal investment with transient price impact.
\newblock {\em SIAM J. Financial Math.}, 10(3):723--768, 2019.

\bibitem{becherer2018optimala}
D.~Becherer, T.~Bilarev, and P.~Frentrup.
\newblock Optimal asset liquidation with multiplicative transient price impact.
\newblock {\em Applied Mathematics \& Optimization}, 78(3):643--676, 2018.

\bibitem{becherer2018optimalb}
D.~Becherer, T.~Bilarev, and P.~Frentrup.
\newblock Optimal liquidation under stochastic liquidity.
\newblock {\em Finance and Stochastics}, 22(1):39--68, 2018.

\bibitem{becherer2019stability}
D.~Becherer, T.~Bilarev, and P.~Frentrup.
\newblock Stability for gains from large investors’ strategies in
  {$M_1$/$J_1$} topologies.
\newblock {\em Bernoulli}, 25(2):1105--1140, 2019.

\bibitem{bertsimas1998optimal}
D.~Bertsimas and A.~W. Lo.
\newblock Optimal control of execution costs.
\newblock {\em Journal of Financial Markets}, 1(1):1--50, 1998.

\bibitem{cheridito2014optimal}
P.~Cheridito and T.~Sepin.
\newblock Optimal trade execution under stochastic volatility and liquidity.
\newblock {\em Applied Mathematical Finance}, 21(4):342--362, 2014.

\bibitem{dolinsky2020note}
Y.~Dolinsky, B.~Gottesman, and O.~Gurel-Gurevich.
\newblock A note on costs minimization with stochastic target constraints.
\newblock {\em Electronic Communications in Probability}, 25, 2020.

\bibitem{fruth2014optimal}
A.~Fruth, T.~Sch\"{o}neborn, and M.~Urusov.
\newblock Optimal trade execution and price manipulation in order books with
  time-varying liquidity.
\newblock {\em Math. Finance}, 24(4):651--695, 2014.

\bibitem{fruth2019optimal}
A.~Fruth, T.~Sch\"{o}neborn, and M.~Urusov.
\newblock Optimal trade execution in order books with stochastic liquidity.
\newblock {\em Math. Finance}, 29(2):507--541, 2019.

\bibitem{gatheral2010no}
J.~Gatheral.
\newblock No-dynamic-arbitrage and market impact.
\newblock {\em Quantitative finance}, 10(7):749--759, 2010.

\bibitem{gatheral2012transient}
J.~Gatheral, A.~Schied, and A.~Slynko.
\newblock Transient linear price impact and {F}redholm integral equations.
\newblock {\em Mathematical Finance: An International Journal of Mathematics,
  Statistics and Financial Economics}, 22(3):445--474, 2012.

\bibitem{graewe2017optimal}
P.~Graewe and U.~Horst.
\newblock Optimal trade execution with instantaneous price impact and
  stochastic resilience.
\newblock {\em SIAM Journal on Control and Optimization}, 55(6):3707--3725,
  2017.

\bibitem{graewe2015non}
P.~Graewe, U.~Horst, and J.~Qiu.
\newblock A non-{M}arkovian liquidation problem and backward {SPDE}s with
  singular terminal conditions.
\newblock {\em SIAM Journal on Control and Optimization}, 53(2):690--711, 2015.

\bibitem{graewe2018smooth}
P.~Graewe, U.~Horst, and E.~S{\'e}r{\'e}.
\newblock Smooth solutions to portfolio liquidation problems under
  price-sensitive market impact.
\newblock {\em Stochastic Processes and their Applications}, 128(3):979--1006,
  2018.

\bibitem{horst2016constrained}
U.~Horst, J.~Qiu, and Q.~Zhang.
\newblock A constrained control problem with degenerate coefficients and
  degenerate backward {SPDE}s with singular terminal condition.
\newblock {\em SIAM Journal on Control and Optimization}, 54(2):946--963, 2016.

\bibitem{horst2019multi}
U.~Horst and X.~Xia.
\newblock Multi-dimensional optimal trade execution under stochastic
  resilience.
\newblock {\em Finance and Stochastics}, 23(4):889--923, 2019.

\bibitem{huberman2004price}
G.~Huberman and W.~Stanzl.
\newblock Price manipulation and quasi-arbitrage.
\newblock {\em Econometrica}, 72(4):1247--1275, 2004.

\bibitem{kruse2016minimal}
T.~Kruse and A.~Popier.
\newblock Minimal supersolutions for {BSDE}s with singular terminal condition
  and application to optimal position targeting.
\newblock {\em Stochastic Processes and their Applications}, 126(9):2554--2592,
  2016.

\bibitem{kyle1985continuous}
A.~S. Kyle.
\newblock Continuous auctions and insider trading.
\newblock {\em Econometrica}, 53(6):1315--1335, 1985.

\bibitem{lorenz2013drift}
C.~Lorenz and A.~Schied.
\newblock Drift dependence of optimal trade execution strategies under
  transient price impact.
\newblock {\em Finance and Stochastics}, 17(4):743--770, 2013.

\bibitem{obizhaeva2013optimal}
A.~A. Obizhaeva and J.~Wang.
\newblock Optimal trading strategy and supply/demand dynamics.
\newblock {\em Journal of Financial Markets}, 16:1--32, 2013.

\bibitem{popier2019second}
A.~Popier and C.~Zhou.
\newblock Second-order {BSDE} under monotonicity condition and liquidation
  problem under uncertainty.
\newblock {\em The Annals of Applied Probability}, 29(3):1685--1739, 2019.

\bibitem{predoiu2011optimal}
S.~Predoiu, G.~Shaikhet, and S.~Shreve.
\newblock Optimal execution in a general one-sided limit-order book.
\newblock {\em SIAM Journal on Financial Mathematics}, 2(1):183--212, 2011.

\bibitem{schied2013control}
A.~Schied.
\newblock A control problem with fuel constraint and {D}awson--{W}atanabe
  superprocesses.
\newblock {\em The Annals of Applied Probability}, 23(6):2472--2499, 2013.

\bibitem{SchSchTeh2010}
A.~Schied, T.~Schöneborn, and M.~Tehranchi.
\newblock Optimal basket liquidation for {CARA} investors is deterministic.
\newblock {\em Applied Mathematical Finance}, 17(6):471--489, 2010.

\bibitem{schied2009risk}
A.~Schied and T.~Sch{\"o}neborn.
\newblock Risk aversion and the dynamics of optimal liquidation strategies in
  illiquid markets.
\newblock {\em Finance and Stochastics}, 13(2):181--204, 2009.

\end{thebibliography}
\end{document}